\documentclass[12pt, draftclsnofoot, onecolumn]{IEEEtran}

\usepackage{cite}      
\usepackage{graphicx}  
\usepackage{psfrag}                            
\usepackage{subfigure} 
\usepackage{url}       
\usepackage{stfloats}  
\usepackage{amsmath}   
\usepackage{array}
\usepackage{times}
\usepackage{epsfig,graphicx,color,pstricks}
\usepackage{amsmath,amssymb,amsthm,amsbsy,amsfonts,latexsym}
\usepackage{bm}
\usepackage{psfrag}
\usepackage[mathscr]{eucal}
\usepackage{cite}
\usepackage[all]{xy}
\usepackage{subfigure}
\usepackage{graphics}

\newtheorem{thm}{Theorem}
\newtheorem{cor}[thm]{Corollary}

\theoremstyle{remark}
\newtheorem{rem}{Remark}
\theoremstyle{definition}

\allowdisplaybreaks
\begin{document}
\title{ Coverage in mmWave Cellular Networks with Base station Cooperation }
\author{
\IEEEauthorblockN{Diana Maamari, Natasha Devroye, Daniela Tuninetti \\}
\IEEEauthorblockA{%
University of Illinois at Chicago, Chicago IL 60607, USA,\\
Email: {\tt dmaama2, devroye, danielat @ uic.edu}}%
 \thanks{%
Diana Maamari, Natasha Devroye, and Daniela Tuninetti are with the Electrical and Computer Engineering Department of the University of Illinois at Chicago, Chicago, IL 60607 USA (e-mail: \{dmaama2, devroye, danielat\}@ uic.edu); their work was partially funded by NSF under award number 1017436,1422511,1216825; the contents of this article are solely the responsibility of the author and do not necessarily represent the official views of the NSF.
The results in this paper will be submitted in part in the 2015 IEEE Global Communications Conference.
}}
\maketitle

\begin{abstract}

The presence of signal outage, due to shadowing and blockage, is expected to be the main bottleneck in millimeter wave (mmWave) networks. Moreover, with the anticipated vision that mmWave networks would have a dense deployment of base stations, interference from strong line-of-sight base stations increases too, thus further increasing the probability of outage. To address the issue of reducing outage, this paper explores the possibility of base station cooperation in the downlink of a mmWave heterogenous network. The main focus of this work is showing that, in a stochastic geometry framework, cooperation from randomly located base stations decreases outage probability. With the presumed vision that less severe fading will be experienced due to highly directional transmissions, one might expect that cooperation would increase the coverage probability; our numerical examples suggest that is in fact the case. Coverage probabilities are derived accounting for: different fading distributions, antenna directionality and blockage. Numerical results suggest that coverage with base station cooperation in dense mmWave systems and with no small scale fading considerably exceeds coverage with no cooperation. In contrast, an insignificant increase is reported when mmWave networks are less dense with a high probability of signal blockage and with Rayleigh fading. 
\end{abstract}

\section{Introduction}
\label{sec:Intro}

One of the fundamental goals for 5G  is a radical increase in data rates~\cite{survey:whatwill5gbe}.  It is anticipated that higher data rates will be achieved by extreme densification of base stations, massive multiple-input-multiple-output (MIMO), increased data rate and/or base station cooperation~\cite{survey:whatwill5gbe}. However, prime microwave wireless spectrum has become severely limited, with little unassigned bandwidth available for emerging wireless products and services. Therefore, to fulfill the need for increased bandwidth, millimeter wave (mmWave) spectrum between 30 and 300 GHz have been considered for future 5G wireless mobile networks. Until recently, mmWave frequency bands were presumed to be unreliable for cellular communication due to blockage, absorption, diffraction, and penetration, resulting in outages and unreliable cellular communications~\cite{elzammWavesurvey}. However, the advances in CMOS radio-frequency circuits, along with the very small wavelength of mmWave signals, allows for the packing of large antenna arrays at both the transmit and receive ends, thus providing highly directional beam forming gains and acceptable signal-to-noise ratio (SNR)~\cite{elzammWavesurvey},~\cite{omarmmWaveprecoding}. This directionality will also lead to reduced interference when compared to microwave networks~\cite{elzammWavesurvey}.  It is thus anticipated that mmWave spectrum holds tremendous potential for increasing spectral efficiency in upcoming cellular systems~\cite{Rappa_itwillwork}.

To further address the demand for higher data rates, cooperation between macro, pico and femto base stations has been proposed to enable a uniform broadband user experience across the network. The dynamic coordination across several base stations - known as coordinated multipoint (CoMP) - will limit the intercell interference thus increasing throughput and enhancing performance at cell borders~\cite{complte}.

\paragraph*{Past Work}

Coverage and capacity in mmWave cellular systems and in CoMP networks have been studied. In~\cite{salammmWavecoverage} the authors compared the performance, in terms of coverage and capacity, of a stochastic geometry based mmWave network (without CoMP) to a microwave cellular network, at a single antenna receiver (typical user). In~\cite{salammmWavecoverage}, directionality at the transmitters, intra-cell and inter-cell interference were accounted for but blockage was not included in the analysis. The authors show that coverage in mmWave systems increases with the decrease in the half-power beam width of the radiation pattern. In fact, having narrower beams decreases beam overlap, thus decreasing intra-cell and inter-cell interference and increasing coverage probability. In this paper, we propose to study the problem of base station cooperation in the downlink of dense mmWave heterogenous network as a means to combat blockage and decrease signal outage. Our derivations of the coverage probability, similarly to~\cite{salammmWavecoverage}, account for interference experienced at the typical user, but in addition blockage is incorporated in the analysis.

In~\cite{taimmWavedenseconference} (see also journal version in~\cite{taimmWavedensejournal}) the authors proposed a stochastic geometry framework to evaluate the performance of mmWave cellular networks (without CoMP) with blockage. The authors incorporate blockage by modeling the probability of a communication link - being either a line-of-sight (LOS) or non-LOS (NLOS) link - as function of the length of the communication link from the serving base station. Different pathloss laws were applied to the LOS and NLOS links. 
Numerical results in~\cite{taimmWavedensejournal} suggest that higher data rates can be achieved when compared to microwave cellular networks. One of the interesting observations made in~\cite{taimmWavedenseconference} is that mmWave networks should be dense but not too dense - since the number of LOS interfering base stations increases when the density of base stations increases. We willl leverage results from~\cite{taimmWavedensejournal} to incorporate blockage and differentiate between having LOS links and NLOS links from the base stations in the analysis of the problem of joint transmission in mmWave networks.

In~\cite{gauravmicrowavecomp} (see also journal version in~\cite{Nigam_CoMP_journal}) the authors used stochastic geometry for studying microwave joint transmission CoMP where single antenna base stations transmit the same data to single antenna users. Different performance metrics (including coverage probability) were considered, to evaluate the performance at the typical user located at an arbitrary location (general user) and receiving data from base stations with the strongest average received power. A user at the cell-corner (worst-case user) was also considered. The coverage probability was derived for both types of users under the assumption that the base stations have no CSI. The case with full CSI was evaluated with different performance metrics (diversity gain and power gain). 
The derivation of the coverage probability for a mmWave network with base station cooperation in this work is similar to that in~\cite{gauravmicrowavecomp} for the general user, except that key factors specific to the mmWave channel model have to incorporated, some of which are the high directional transmission at the base stations, blockage and improved fading distribution due to sparse scattering.  

In~\cite{omarmmWaveprecoding} the authors considered the problem of finding a suitable single user MIMO transmit precoding and receive combining in mmWave systems under a set of hardware constraints suitable for large antenna arrays. Both problems (transmit precoding and receiver combining) were formulated as a sparsity constrained signal recovery problem and solved using orthogonal matching pursuit algorithms. The solution suggests that the transmitter applies a number of array response vectors at the RF level (which are phase only vectors) and forms linear combinations of these vectors using a digital precoder.  A similar observation was made for the receiver combining operation.
In~\cite{Caire_JSDM1} a multiuser MIMO downlink scheme, Joint Division Spatial Multiplexing (JSDM), was proposed. The scheme is suitable for frequency division duplexing (FDD) systems with large number of antennas (massive MIMO) and non-ideal channel state information (CSI) at the base station. The base stations equipped with multiple antennas were assumed to serve $K$ single-antenna users. Users who have identical covariance matrices were grouped together while separate groups of users were assumed to have almost orthogonal eigenspaces of channel covariance matrices.  
The proposed two-stage JSDM exploits channel matrices properties and finds the optimal precoding and prebeamforming matrices. 
In this paper, we assume that the cooperating base stations beam steer to the typical user using vectors that can be readily implemented using phase shifters, and the receiver applies a single vector to process the received signal from the cooperating base stations.

%
\paragraph*{ Main Contributions}
\label{sec:maincontribution}
In this paper, we propose to study the benefits of base station cooperation in the downlink of a heterogenous mmWave cellular system as a mean to decrease signal outage. We anticipate that the cooperation provides substantial gain in coverage with the anticipated improved fading distribution and extreme base station densification in mmWave networks.
Our extensive numerical examples show that this is in fact the case for the following scenarios, Case 1) for dense mmWave networks where the number of interfering LOS base stations increases and Case 2) when there is no small scale fading channel on the channel gains from the cooperating base stations (a good assumption due to the high directional transmission). We also provide examples when cooperation does not provide substantial increase in coverage probability.
We consider a stochastic geometry based model as in~\cite{salammmWavecoverage,taimmWavedenseconference, taimmWavedensejournal,gauravmicrowavecomp,jefftractablemodel,Nigam_CoMP_journal}, to study coverage in CoMP heterogenous mmWave network. To do so we need to incorporate key factors specific to a mmWave channel model. These specific mmWave characteristics are: a realistic mmWave channel model, highly directional channel gains and sensitivity to blockages. 
%

Coverage probabilities are derived for the case of a single antenna receiver (typical user). We use concepts from~\cite{salammmWavecoverage,taimmWavedenseconference,taimmWavedensejournal} to incorporate blockage, interference and different fading distributions (Rayleigh, Nagakami and no fading) in our analysis. The joint distribution of the cooperating base stations to the typical user in the presence of blockage is also derived. 
%
%
%
%

\paragraph*{Paper Organization and Notations}
\label{sec:paperorganizationaandnotations}
The downlink CoMP mmWave heterogenous network model, the beam steering at the base stations and the decoding at the typical user are explained in Section~\ref{sec:network model}. The coverage probability in the absence of blockage, and with Rayleigh fading is derived in Section~\ref{sec:performance analysis no blockage ULA analysis}. In Section~\ref{sec:coverageprobabilitywithblockage}, we consider Rayleigh fading mmWave networks with a blockage parameter at each tier, and the coverage probability is derived accordingly. In Section~\ref{sec:Coverage Probability with Nakagami fading and blockage}, we derive the coverage probability for the same network model with blockage but use the Nakagami fading distribution to model the fading distribution on the direct links of the cooperating base stations. The assumption of having no small scale fading for the channel gains from the cooperating base stations is further considered in Section~\ref{sec:CoverageProbabilitywithNoSmallScaleFading}.
Proofs may be found in the Appendices.  Tables~\ref{Table1},~\ref{Table2},~\ref{Table3},~\ref{Table4} summarize all the notations used throughout the paper. 
\begin{table*}[h]
\centering
\caption{Poisson point process variables} 
\label{Table1}
\scalebox{0.85}{
\begin{tabular}{| l | l |}
\hline
Notation  & Description  \\ 
\hline
$K$& Total number of tiers \\
\hline
$\Phi_{k}$& Homogenous Poisson Point Process (PPP) indexed by $k\in[1:K]$ \\
\hline 
$\lambda_k$ & Intensity of the PPP $\Phi_{k}$\\
\hline
$P_k$ & Available power at each base station that belongs to tier $k\in[1:K]$\\
\hline 
$ v$& Points on 2D plane representing location of base stations\\
\hline 
$\| v\| $ & Distance from point $v$ to the typical user located at the origin\\
\hline 
$\alpha$ & Pathloss exponent assumed equal for all tiers\\
\hline
$\Theta_{k}=\{\frac{\| v\| ^\alpha}{P_k}, v\in \Phi_{k}\}$ & Normalized pathloss between each base station in $\Phi_{k}$ and the typical user\\
\hline 
$\lambda_k(v)$ & Intensity of $\Theta_k$\\
\hline
$\Theta=\cup_{k=1}^K \Theta_k $ & Process representing the union of non-homogenous PPP, elements are indexed in increasing order WLOG \\
\hline 
$\lambda(v)= \sum_{k=1}^K\lambda_k(v)$ & Intensity of $\Theta$\\
\hline
$\gamma'_i= \frac{\| v_i\| ^\alpha}{P_k}$ & Normalized pathloss\\
\hline
${\bf \gamma'}=\{\gamma'_{1},\cdots,\gamma'_{n}\} $ & Set of normalized pathloss of the cooperating base stations\\
\hline
$ f_{\Gamma'}(\gamma')$ & Joint distribution of $\gamma'$\\
\hline
\end{tabular}}
\end{table*}

\begin{table*}[h]
\centering
\caption{General channel model variables} 
\label{Table2}
\scalebox{0.85}{
\begin{tabular}{| l | l |}
\hline
Notation  & Description  \\ 
\hline 
$N_t, N_r$ & Number of antennas at each base station and at the receiver\\
\hline
$H_v$ & MIMO channel from base station at location $v$ to typical user\\
\hline 
$h_v$ &Small scale fading \\
\hline 
$L_v$ & Number of channel clusters \\
\hline
$\phi_v^t$ & Path angle at the transmitter\\
\hline 
$\phi_v^r$ & Path angle at the receiver\\
\hline
$f(v)$ & Function that returns the index to which a base station at $v$ belongs to\\
\hline
$\gamma_v= \frac{P_{f(v)}}{\| v\| ^\alpha}$ & Pathloss\\
\hline 
$\bf{a_{t(r)}}(.)$ & Uniform linear array vector representation at the transmitter (receiver)\\
\hline 
$\Delta_{t(r)}$ & Normalized transmit (receive) antenna separation\\
\hline 
$L_t$ & Normalized length of the transmit antenna array\\
\hline 
$\bf{n}$ & Noise vector of i.i.d $\mathcal{CN}(0,\sigma_n^2)$\\
\hline 
\end{tabular}}
\end{table*}

\begin{table*}
\centering
\caption{Channel variables from the cooperating base stations} 
\label{Table3}
\scalebox{0.85}{
\begin{tabular}{| l | l |}
\hline
Notation  & Description  \\ 
\hline 
$\mathcal{T}$ & Set of cooperating base stations with cardinality $|\mathcal{T}|=n$\\
\hline
$v_i,\ \ i\in[1:|\mathcal{T}|]$ & Points on the 2D plane corresponding to cooperating base stations location (sometimes indexed by $j$ instead of $i$)\\
\hline 
${\bf H}_{v_i}$ & MIMO channel from the cooperating base stations\\
\hline 
$h_{v_i}, \phi_{v_i}^t, \phi_{v_i}^r, \gamma_{v_i}$ &  Channel parameters of the interfering links as defined in Table~\ref{Table2} \\
\hline 
$\Omega_{\phi_{v_i}^r}$ & Directional cosine given by cos$(\phi_{v_i}^r)$\\
\hline
${\bf{X}}_{v_i}$ & Transmit signal from cooperating base stations\\
\hline
\end{tabular}}
\end{table*}

\begin{table*}
\centering
\caption{Channel variables from the interfering base stations} 
\label{Table4}
\begin{tabular}{| l | l |}
\hline
Notation  & Description  \\ 
\hline
$l_i,\ \ i\in[1:|\mathcal{T}^c|]$ & Points on the 2D plane corresponding to interfering base stations locations\\
\hline 
${\bf H}_{l_i}$ & MIMO channel from the interfering base stations\\
\hline
$h_{l_i}, \phi_{l_i}^t, \phi_{l_i}^r, \gamma_{l_i}$ & Channel parameters of the interfering links as defined in Table~\ref{Table2} \\
\hline
$\Omega_{\phi_{l_i}^r}$ & Directional cosine given by cos$(\phi_{l_i}^r)$\\
\hline
$\theta_{l_i}^t$ & Angle used by interfering base station at position $l_i$ to beam steer to a user other than typical user\\
\hline
$\Omega_{\theta_{l_i}^t}$ & Directional cosine given by cos$(\theta_{l_i}^t)$\\
\hline
${\bf{X}}_{l_i}$ &Transmit signal from interfering base stations\\
\hline
\end{tabular}
\end{table*}


\section{Coverage Probability with no Blockage}
\label{sec:network model}
\subsection{Network Model} 
Consider a $K$ tier heterogenous network where each tier is an independent two-dimensional homogenous Poisson point process (PPP). We denote the base station location process of tier $k \in [1:K]$ by $\Phi_{k}$ with density $\lambda_k$. The mmWave base stations that belong to the same tier $k$ transmit with the same power $P_k$ for $k \in [1:K]$. We study the coverage probability as experienced by the typical user located at the origin, and denote the set of cooperating base stations, which jointly transmit to the typical user, by $\mathcal{T} \in \cup_{k=1}^K \Phi_{k}$. We assume that $|\mathcal{T}|=n$, and that these $n$ base-stations correspond to those with the strongest received power at the typical user receiver. 
In the rest of the section, we first describe the channel model and then derive the output signal at the typical user receiver.
\subsection{Simplified Clustered Channel Model}
\label{sec:Clustered channel model}
A clustered channel model,~\cite{omarmmWaveprecoding},~\cite{mmWavemustapha}, is used to model the wireless channel between the base stations and the typical user located at the origin. We assume all base stations have the same number of transmit antennas $N_t$, while the receiver has $N_r$ receive antennas. 
The $N_r \times N_t$ channel matrix ${\bf H}_{v}$, between a base station located at $v \in \mathbb{R}^2$ and the typical user is the sum of $L_{v}$ clusters and is expressed as
\begin{align}\label{eq:channel matrix}
{\bf H}_{v}= \frac{\sqrt{N_t N_r}}{L_{v}} \sum_{l=1}^{L_{v}} \sqrt{ \gamma_{v,l}} h_{v,l}{\bf a_r}(\phi_{v,l}^r) \ {\bf a_t}(\phi_{v,l}^t)^*,
\end{align} 
where 
 \begin{itemize}
 \item $\gamma_{v,l}= \frac{P_{f(v)}}{{\| v\| }^{\alpha} }$ is the pathloss,
 \item $f(v)$ is a function that returns the index $k$ of the tier to which the base station at location $v$ belongs to, 
 \item $\alpha$ is the pathloss exponent, 
 \item $\| v\| $ is the distance from the base station at location $v$ to the user at the origin,
 \item $h_{v,l}$ is the complex fading channel gain, 
 \item The vectors ${\bf a_{t(r)}}(\phi_{v}^{t(r)})$ are the normalized uniform linear array (ULA) transmit and receive array response and are given by~\cite[Eq. (7.21), Eq. (7.25)]{book:davidtse} 
\begin{align}
{\bf a_{t(r)}}(\phi_{v}^{t(r)})= \frac{1}{\sqrt{N_{t(r)}}}[1, e^{- jA}, e^{- j2A},\cdots, e^{- j(N_{t(r)}-1)A}]^T 
\end{align}
where $A= 2\pi \Delta_{t(r)}  \cos({\phi_{v}^{t(r)})}$ and $\Delta_{t(r)}$ is the normalized transmit (receive) antenna separation (normalized to the unit of the carrier wavelength), at a path angle $\phi_{v}^{t(r)}$ of departure (arrival) from the base station $v$.
\end{itemize}

In the following, for simplicity, we shall consider the case $L_v=1$.

\subsection{Received Signal at the Typical User}
\label{sec:receivedsignalatthetypicaluser}
In this section we will further divide the points $v \in \mathbb{R}^2$ into a set of points $v_i$ and $l_i$ to differentiate between the location of the cooperating and interfering base stations, respectively. The $N_r \times N_t$ desired channel matrices are denoted by ${\bf H}_{v_i}$ for $i \in [1:|\mathcal{T}|=n]$, where $n$ is a positive constant, while the interfering channel matrices are denoted by ${\bf H}_{l_i}$ for $i\in [1:|\mathcal{T}^c|]$. 

\begin{figure}
\centering
\includegraphics[width=10cm]{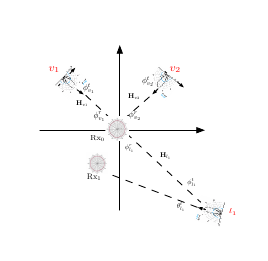}
\caption{A typical user is served by two cooperating base stations at locations $v_1$ and $v_2$, while being interfered by base station at location $l_1$.}
\label{fig:networkmodel}
\end{figure}

Fig.~\ref{fig:networkmodel} shows an example of a network model, where two base stations at locations $v_1$ and $v_2$, jointly transmit to the typical receiver located at the origin (indicated as Rx$_{0}$) in the presence of a single interfering base station at location $l_1$. The MIMO channel matrices between the cooperating base stations and the typical user are given by ${\bf H}_{v_1}$ and ${\bf H}_{v_2}$. The channel matrix between the interfering base station and the typical user is denoted by ${\bf H}_{l_1}$. The angles,  $\phi_{v_i}^t$ and $\phi_{v_i}^r$, are the cluster's angle of departure and arrival respectively from the base station $v_i$, $i\in[1:2]$, to the typical receiver. The angle $\phi_{l_1}^t$ is the angle of departure of the cluster from the interfering base station. The base station at $l_1$ uses a beam steering angle $\theta_{l_1}$ to transmit data to some other user (not the typical user) indicated as Rx$_{1}$. The received signal is
\begin{align} \nonumber
{\bf y} &= \sum_{i=1}^{|\mathcal{T}|=n}{\bf H}_{v_i}{\bf {X}}_{v_i} + \sum_{i=1}^{|\mathcal{T}^c|}{\bf H}_{l_i}{\bf {X}}_{l_i}+{\bf n} \\ 
&=\sum_{i=1}^{|\mathcal{T}|=n}\sqrt{N_tN_r}\sqrt{\gamma_{v_i}} h_{v_i}  \ {\bf a_r}(\phi_{v_i}^r) \ {\bf a_t}(\phi_{v_i}^t)^* {\bf X}_{v_i} +\sum_{i=1}^{|\mathcal{T}^c|}\sqrt{N_tN_r}\sqrt{\gamma_{l_i}} h_{l_i}  \ {\bf a_r}(\phi_{l_i}^r) \ {\bf a_t}(\phi_{l_i}^t)^*{\bf X}_{l_i}+{\bf n}\label{eq:received signal}
\end{align}
where the first sum in~\eqref{eq:received signal} is the desired signal from the mmWave cooperating base stations, while the second sum contains the signals from the interfering base stations.

The user associates with a set of cooperating base stations $\mathcal{T}$, that provide the strongest average received power as in~\cite{Nigam_CoMP_journal}. Specifically,
 \begin{align}
\mathcal{T} = \arg\max_{\{v_1\cdots v_n\} \subset \cup_{k=1}^{K} \Phi_{k}} \sum_{i=1}^n \frac{P_{f(v_i)}}{{\| v_{i}\| }^{\alpha} }
\end{align}
and $\mathcal{T}^c:= \cup_{k=1}^K \ \Phi_{k} \backslash \ \mathcal{T}$. The path angles, $\phi_{v_i}^{t}$, $i\in[1:|\mathcal{T}|]$, and $\phi_{l_i}^{t}$, $i\in[1:|\mathcal{T}^c|]$, represent the angle of departure of the desired and interfering paths respectively, while $\phi_{v_i}^{r}$, $i\in[1:|\mathcal{T}|]$, and $\phi_{l_i}^{r}$, $i\in[1:|\mathcal{T}^c|]$, represent the angle of arrival of the received path from the cooperating and interfering base stations respectively. The transmit signals, ${\bf X}_{v_i}$, $i\in[1:|\mathcal{T}|]$ and ${\bf X}_{l_i}$, $ i \in [1:|\mathcal{T}^c|]$, represent the signal from the cooperating and interfering base stations within $\mathcal{T}$ and $\mathcal{T}^c$ respectively. ${\bf n}$ is the noise vector of i.i.d $\mathcal{CN}(0, \sigma^2)$ components.
\subsection{Beam steering}
The base stations in $\mathcal {T}$ jointly send the same data to the receiver. Each base station beam steers to the typical user, therefore the transmitted signal is
\begin{align}\label{eq:transmitted signal}
{\bf X}_{v_i}={\bf a_t}(\phi_{v_i}^t)s 
\end{align}
for $i \in[1:|\mathcal{T}|=n]$, where $s$ is channel input symbol transmitted by the cooperating base stations to the typical receiver. The signals transmitted by the interfering base stations are
\begin{align}
{\bf X}_{l_i}= {\bf a_t}(\theta_{l_i}^t)s_{l_i} 
\end{align}
for $i \in [1:|\mathcal {T}^c|]$, where $s_{l_i}$ is the channel input symbol transmitted by the interfering base stations, while the angle $\theta_{l_i}^t$ is the angle used by base station $l_i$ to beam steer to a user other than the typical user, and is different from $\phi_{l_i}^t$ in general. We assume that $s$ and $s_{l_i}$ are independent zero mean and unit variance random variables.      
  
{\it Assumption 1:} We assume that the cooperating base stations have perfectly beam steered to the typical receiver: notice that the angles in~\eqref{eq:transmitted signal}, used by the base station to beam steer, are equal to the clusters' angles of departure in the desired channel in~\eqref{eq:received signal}.
\subsection{Decoding}
\label{sec:decoding}
The receiver uses a single vector ${\bf w}\in \mathbb{C}^{N_r\times 1}$ to detect the scalar transmit symbol, that is, the processed received signal is given by
\begin{align}
&{\hat y }= {\bf w^*} {\bf y}\\
&{\bf w}= \sum_{j=1}^{n}{\bf a_r}(\phi_{v_j}^r) 
\label{eq:decoder}
\end{align}

\begin{rem}
The choice of ${\bf w}$ in~\eqref{eq:decoder} is one choice of a decoder that can be implemented readily using phase shifters in the RF domain (analog processing), in fact if one wants to consider a near optimal performance, then the work in~\cite{omarmmWaveprecoding}, which finds a hybrid MIMO receiver combining algorithm and minimizes the mean-square-error between the transmitted and received signals under a set of RF hardware constraints for the resulting point-to-point channel should be generalized to finding a suitable algorithm for the downlink cooperative channel. 
\end{rem}

\noindent
{\it Assumption 2: }We assume perfect CSI of the path angles at the decoders since these angles vary slowly. However, we assume that the phases of the complex channel gains, $h_{v_i}, i \in [1:|\mathcal{T}|]$, are not available at the terminals as they change very quickly on the order of a wavelength and thus cannot be tracked. 
The performance here should be considered as an upper bound on the performance of the more realistic case with imperfect path angle assumption.
\subsection{Output Signal}
\label{sec:outputsignal}
The output signal at the typical user under the previously stated assumptions is given by
\begin{subequations}
\begin{align}
\hat{y}&= {\bf w^*} {\bf y}=\sum_{j=1}^{n}{\bf a_r}(\phi_{v_j}^r)^*\nonumber
\biggl(\sqrt{N_tN_r}\sum_{i=1}^{n}\sqrt{\gamma_{v_i}} h_{v_i}{\bf a_r}(\phi_{v_i}^r ) {\bf a_t}(\phi_{v_i}^t)^*{\bf a_t}(\phi_{v_i}^t)s\\& +\sqrt{N_tN_r} \sum_{i=1}^{ |\mathcal T^c| }\sqrt{\gamma_{l_i}} h_{l_i} {\bf a_r}(\phi_{l_i}^r) {\bf a_t}(\phi_{l_i}^t)^* {\bf a_t}(\theta_{l_i}^t) s_{l_i} \biggl)+{z}
\\&=\nonumber
 \sqrt{N_tN_r}\sum_{j=1}^{n}\sum_{i=1}^{n} \sqrt{\gamma_{v_i}} h_{v_i}
 G_r(\Omega_{\phi_{v_j}^r}-\Omega_{\phi_{v_i}^r})
 G_t(\Omega_{\phi_{v_i}^t}-\Omega_{\phi_{v_i}^t})
 s 
\\&+
\sqrt{N_tN_r}\sum_{j=1}^{n}\sum_{i=1}^{ |\mathcal T^c| }\sqrt{\gamma_{l_i}} h_{l_i}
 G_r(\Omega_{\phi_{v_j}^r}-\Omega_{\phi_{l_i}^r}  )
 G_t(\Omega_{\phi_{l_i}^t}-\Omega_{\theta_{l_i}^t})
 s_{l_i} 
+{z}
\label{eq:receivedsignal}
\end{align}
\end{subequations}
where { $z= {\bf w^* n} \sim \mathcal{C}N(0,\sigma_n^2)$}, with $\sigma_n^2=\sigma^2 {\bf w^*}{\bf w}$ and where we introduced the antenna-array-gain functions
\begin{align}\label{eq:antennagainfunction}
G_x(y) := 
{\rm e}^{j \pi \Delta_x (N_x-1) y} \
\frac{\sin (\pi \Delta_x N_x y)}
 {N_x \sin (\pi \Delta_x     y)}: |G_x(y)|\leq 1, \quad x\in\{t,r\} ,
\end{align}
\begin{align}
{\bf a}_x(\phi_1)^* {\bf a}_x(\phi_2)=G_x(\Omega_{\phi_1}-\Omega_{\phi_2}), \quad x\in\{t,r\},
\end{align}
with
$\Omega_{\phi} := \cos(\phi)$ and $\Delta_x$, $x\in\{t,r\}$ being the normalized antenna separation.
\subsection{ SINR Expression}
\label{sec:gainfunctionsandSINRexpression}
Based on~\eqref{eq:receivedsignal}, the instantaneous SINR is then given by
\begin{align}\label{eq:SINRexpressionfullygeneral}
\text{SINR} = \frac{\big|\sum\limits_{i=1}^n\sqrt{\gamma_{v_i}} h_{v_i}C_{v_i}|^2}
{\frac{\sigma_n^2}{N_tN_r}
+\sum\limits_{i=1}^{|\mathcal{T}^c|} \gamma_{l_i} |h_{l_i}|^2 |D_{l_i}|^2\bigl |
G_t(\Omega_{\phi_{l_i}^t}-\Omega_{\theta_{l_i}^t})
|^2},\end{align}
where $C_{v_i}:= \sum_{j=1}^{n}G_r(\Omega_{\phi_{v_j}^r}-\Omega_{\phi_{v_i}^r})$ and $D_{l_i}:= \sum_{j=1}^{n}G_r(\Omega_{\phi_{v_j}^r}-\Omega_{\phi_{l_i}^r})$.

Assuming a single antenna receiver with $N_r=1$ ({$C_{v_i}=D_{l_i}=n$ and $\sigma_n^2= n^2 \sigma^2$), the SINR in~\eqref{eq:SINRexpressionfullygeneral} simplifies to
\begin{align}
\label{eq:SINRexpressionomnicase}
\text{SINR} 
&= \frac{\big|\sum\limits_{i=1}^n\sqrt{\gamma_{v_i}} h_{v_i}|^2}
{\frac{\sigma^2}{N_t}
+\sum\limits_{i=1}^{|\mathcal{T}^c|} \gamma_{l_i} |h_{l_i}|^2 \bigl |
G_t(\Omega_{\phi_{l_i}^t}-\Omega_{\theta_{l_i}^t})
|^2},
\end{align}

The coverage probability for the typical user with SINR as in~\eqref{eq:SINRexpressionomnicase} will be derived under the assumption that all angles are independent and uniformly distributed between $[-\pi,+\pi]$. We will first assume that the receiver is present in a rich scattering environment (Rayleigh fading assumption), and in this scenario the coverage probability is given in Th.~\ref{theorem:coverageprobabilityULAanalysis}. The case where each tier experiences blockage is then considered and the coverage probability is derived accordingly and is given in Th.~\ref{theorem:coverageprobabilitywithblockage}. The Nakagami fading distribution is then used to model the less severe fading distribution on the direct cooperating links and two upper bounds on the coverage probability are then derived and are given in Th.~\ref{theorem:coverageprobabilitywithblockageNakagamifading} and Corollary~\ref{corr: SNRnakagmi}. The case where there is no small scale fading for the direct cooperating links is then considered and the coverage probability for this case is given in  Th.~\ref{thm:nofadingcaseSNRcase}. Future work includes deriving the coverage probability for all the different cases described above using~\eqref{eq:SINRexpressionfullygeneral}, i.e, multiple antennas at the receivers.

\subsection{Performance Analysis}
\label{sec:performance analysis no blockage ULA analysis}
\begin{thm}\label{theorem:coverageprobabilityULAanalysis}
The coverage probability for the typical user, with a single antenna, in a downlink mmWave heterogenous network with $K$ tiers, with base stations having ULA with $N_t$ antennas, of which $n$ jointly transmit to it, is given by
\begin{align} \label{eq:coverageprobabilityULAanalysis}
\normalfont \mathbb{P}({\text{SINR} > T})= \int\limits_{0<\gamma'_1<\cdots<\gamma'_n<+\infty} \mathcal{L}_I\bigg(\frac{T}{\sum_{i\leq n}\gamma_{i}'^{-1}}\bigg)\mathcal{L}_N\bigg(\frac{T }{\sum_{i\leq n}\gamma_{i}'^{-1}}\bigg) f_{\Gamma'}(\gamma')d\gamma'
\end{align}
where  $\gamma_{i}'=\frac{\| v_i\| ^\alpha}{P_{f(v_i)}}$ for  $i \in [1:n]$, and the Laplace transform of the interference and the noise are given by 
\begin{align}\label{eq:Laplaceinterferencecase1}
\mathcal{L}_I(s)= \normalfont \text{exp} \Bigg( -\int_{\gamma_n'}^{\infty} \Bigg[ 1- &\int_{-2}^{+2}\Bigg(\frac{1}{1+s |G_t(\varepsilon)|^2 v^{-1} }\biggl) f_{\Upsilon}(\varepsilon)\ d\varepsilon \Bigg] \ \lambda(v) \, dv \Bigg),\\
&\mathcal{L}_{N}(s) =e^{-s\sigma^2/N_t}, \label{eq:Laplaceofnoise}
\end{align} 
where the antenna array gain $G_t(\varepsilon)$ is given by~\eqref{eq:antennagainfunction} and the probability density function of $\Upsilon_i=\Omega_{\phi_{l_i}^t}-\Omega_{\theta_{l_i}^t}$ is 
 \begin{align}f_{\Upsilon}(\varepsilon)= \int_{\max\{-1,-1-\varepsilon\}}^{\min\{1,1-\varepsilon\}} \left( \frac{1}{\pi^2\sqrt{1-(\varepsilon+y)^2}} \frac{1}{\sqrt{1-y^2}}\right) \, dy, \end{align}
and  $f_{\Gamma'} (\gamma')$ is the joint distribution of $\gamma'=[\gamma_{1}',\cdots, \gamma_{n}'] $ and is given by 
\begin{align}\label{eq:distributionofnbasestations}
f_{\Gamma'} (\gamma') =\prod_{i=1}^n \lambda(\gamma_i')e^{-\Lambda(\gamma_n')}
\end{align}
while the intensity and intensity measure are given by
\begin{align}
&\lambda(v)= \sum_{k=1}^K \lambda_k\frac{2\pi}{\alpha}P_{k}^{\frac{2}{\alpha}}v^{\frac{2}{\alpha}-1}, \\
&\Lambda(\gamma_n')= \sum_{k=1}^K \pi \lambda_k P_k \gamma_n'^{\frac{2}{\alpha}}
\end{align}
\end{thm}
\begin{proof}
Please refer to Appendix~\ref{appendix:coverageprobabilityULAanalysis} for the proof. 
\end{proof}

\subsection{Numerical Results}
\label{sec:numerical results no blockage ULA analysis}
{\bf Example 1:}
In this section we numerically evaluate Th.~\ref{theorem:coverageprobabilityULAanalysis}. We compute the coverage probability for the typical user in a mmWave CoMP heterogenous network and compare it to the case with no base station cooperation. We consider a two tier network, $K=2$ with parameters given in Table~\ref{table:TableNumericalExample1}. 
The noise variance is given by $\sigma^2(\text{dBm})= -174+10\log_{10}(\text{BW})+ \text{NF (dB)}$, { where BW and NF are abbreviations for bandwidth and noise figure respectively}. In Fig.~\ref{fig:noblockageexample} the coverage probability in~\eqref{eq:coverageprobabilityULAanalysis} for $n=2$ and $n=1$ is plotted. In the absence of blockage, the numerical results show that the increase in coverage probability with cooperation for the case of $N_t=8,16$ antennas is almost  $11\%$ at $T=5$ dB. While the increase is $10 \%$ for $N_t=32,64$ at $T=10$ dB.
As expected, an increase in the number of antennas at the base stations increases the coverage probability. For example, for the same threshold $T=10$ dB, the coverage probability with cooperation and with $N_t=16$ is approximately 0.5 while for $N_t=32$ is 0.65. The increase in coverage probability can be interpreted as follows: the mmWave tier is relatively denser than that of a microwave tier, therefore the number of interfering base stations increases too. Thus, with cooperation limits interference and the coverage probability consequently increases.
\begin{figure}
\centering
\includegraphics[width=10cm]{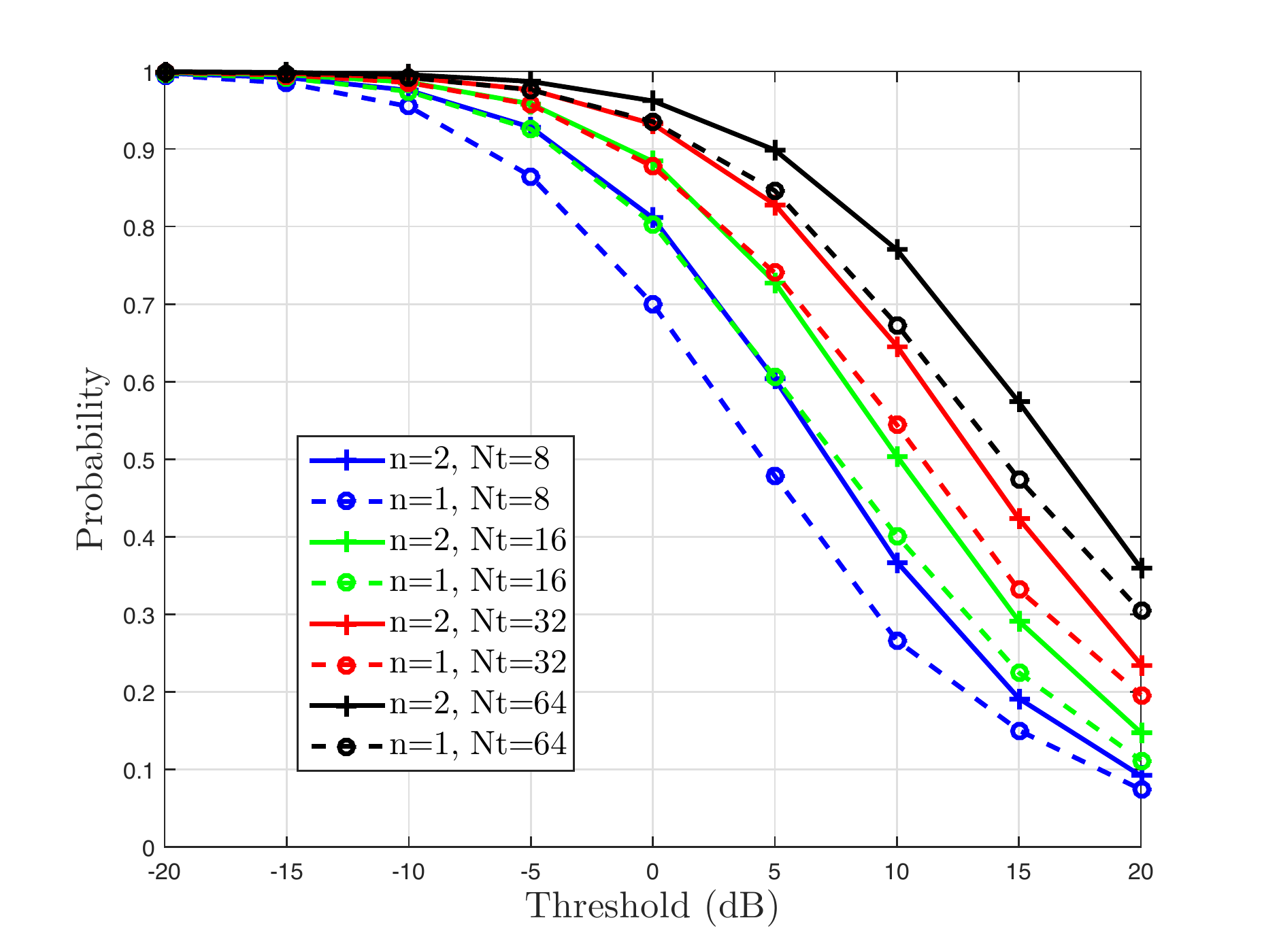}
\caption{Coverage probability in~\eqref{eq:coverageprobabilityULAanalysis} for a two-tier network with parameters in Table~\ref{table:TableNumericalExample1} with two cooperating base stations $(n=2)$ and without base station cooperation $(n=1)$ and for different number of antennas.}
\label{fig:noblockageexample}
\end{figure}

\begin{rem}
The authors in~\cite{gauravmicrowavecomp} compare a two tier network with parameters (power, noise, intensities) suitable for microwave deployment with and without base station cooperation ($n$=2,3); an increase of 17 \% was noted at a threshold $T=0$ dB for the case of CoMP with two cooperating base stations when compared to the case of no cooperation. We shall show that a comparable gain (16\%-18\%) to the one reported in~\cite{gauravmicrowavecomp} can be attained with two cooperating base stations with non fading channel gains in Section~\ref{sec:coveragenofading}. Thus, the Rayleigh fading assumption considered here provides a worse case scenario.
\end{rem}

\begin{table*}
\parbox{.45\linewidth}{
\centering
\caption{Tier for Fig.~\ref{fig:noblockageexample} (Example 1)} 
\label{table:TableNumericalExample1}
\begin{tabular}{| l | l |}
\hline
Parameter  & Value  \\ 
\hline 
Intensity& $\lambda_1= (150^2\pi)^{-1}$, $\lambda_2=(50^2\pi)^{-1}$\\
\hline
Power &$P_1=1$ W (30 dBm) and $P_2=0.25$ W\\
\hline 
Path Loss& $\alpha=3$ \\
\hline 
 Antennas & $N_t =8, 16, 32, 64$\\
\hline 
Noise Figure (NF) & 10 dB\\
\hline
Blockage & Not Applicable\\
\hline
Bandwidth & 1 GHz\\
\hline
\end{tabular}
}
\hfill
\parbox{.45\linewidth}{
\centering
\caption{Tier parameters for Fig.~\ref{fig:blockageexample1} (Example 2)} 
\label{table:TableNumericalExamplewithblockage1}
\begin{tabular}{| l | l |}
\hline
Parameter  & Value  \\ 
\hline 
Intensity& $\lambda= (80^2\pi)^{-1}$\\
\hline
Power &$P_1=1$ W \\
\hline 
Path Loss& $\alpha_1=2, \alpha_2=4$ \\
\hline 
 Antennas & $N_t =16$\\
\hline 
Noise Figure & 5 dB\\
\hline
Blockage & $\beta=0.006, 0.003, 0.0143$\\
\hline
Bandwidth & 1 GHz\\
\hline
\end{tabular}}
\end{table*}



\section{Coverage probability with Blockage}
\label{sec:coverageprobabilitywithblockage}
\subsection{Network model} 
\label{sec:network model blockage}
In this section we again consider a $K$-tier heterogenous network where each tier is an independent two-dimensional homogenous (PPP). The base station location process of each tier is denoted by $\Phi_{k}$ with density $\lambda_k$ for $k \in [1:K]$. Each tier is characterized by a non-negative blockage constant $\beta_k$ for $k\in[1:K]$ { (determined by the density and average size of objects within the tier and where the average LOS range in a tier $k\in[1:K]$ is consequently given by $1/\beta_k$)} as defined in~\cite{taijournalrandomshapetheory} and used in~\cite{taimmWavedenseconference},~\cite{taimmWavedensejournal}. Consequently, after defining the parameter $\beta_k$ for $k\in[1:K]$, we have that the probability of the communication link being a LOS link (no blockage on the link) within tier $k$ is $\mathbb{P}(\text{LOS}_k)= e^{-\beta_k r}$, where $r$ represents the length of the communication link, while the probability of a link being NLOS is $\mathbb{P}(\text{NLOS}_k)=1-\mathbb{P}(\text{LOS}_k)$. The LOS and NLOS links will have different pathloss exponents, $\alpha_1$ and $\alpha_2$, respectively, and are the same for all $k\in[1:K]$. 
With the assumption of blockage the Laplace transform of the interference in~\eqref{eq:Laplaceinterferencecase1} and the joint distribution of the cooperating base stations in~\eqref{eq:distributionofnbasestations} have to be re-derived.


\subsection{Performance Analysis}
\label{sec:performance analysis with blockage}
\begin{thm}\label{theorem:coverageprobabilitywithblockage}
The coverage probability for the typical user, with a single antenna, in a downlink mmWave heterogenous network with $K$ tiers, and where each tier has a blockage parameter $\beta_k$, with $n$ base stations having ULA with $N_t$ antennas, jointly transmitting to it is given by~\eqref{eq:coverageprobabilityULAanalysis},~\eqref{eq:Laplaceinterferencecase1},~\eqref{eq:Laplaceofnoise}
but where now the intensity $\lambda(v)$ in~\eqref{eq:Laplaceinterferencecase1} is given by 
\begin{align}\label{eq:intensitywithblockage}
\lambda(v)=\sum_{k=1}^KA_kv^{\frac{2}{\alpha_1}-1} e^{-a_kv^{\frac{1}{\alpha_1}}} +B_kv^{\frac{2}{\alpha_2}-1}(1-e^{ -b_k v^{\frac{1}{\alpha_2}}})
\end{align}
where $A_k= \pi \lambda_k \frac{2}{\alpha_1}P_k^{\frac{2}{\alpha_1}}, a_k=\beta_kP_k^{\frac{1}{\alpha_1}}, b_k=\beta_kP_k^{\frac{1}{\alpha_2}}, B_k=\pi \lambda_k \frac{2}{\alpha_2}P_k^{\frac{2}{\alpha_2}}$, and the distribution of the distance of $n$ closest base stations is given by~\eqref{eq:distributionofnbasestations} but where 
\begin{align} \nonumber
&\Lambda(\gamma_n')= \sum_{k=1}^K\frac{2\pi \lambda_k}{\beta_k^2}\left(1- e^{-\beta_k ( \gamma_n'P_k)^{\frac{1}{\alpha_1}}}(1+\beta_k ( \gamma_n'P_k)^{\frac{1}{\alpha_1}})\right) \\&\hspace{3cm}+ \pi \lambda_k ( \gamma_n'P_k)^{\frac{2}{\alpha_2}}-\frac{2\pi \lambda_k}{\beta_k^2}\left(1- e^{-\beta_k ( \gamma_n'P_k)^{\frac{1}{\alpha_2}}}(1+\beta_k ( \gamma_n'P_k)^{\frac{1}{\alpha_2}})\right)\label{eq:intensitymeasurewithblockage}
\end{align}
\end{thm}
\begin{proof}
Please refer to Appendix~\ref{appendix:coveragewithblockage} for detailed proof.
\end{proof}
\subsection{Numerical Results}
In this section we numerically evaluate Th.~\ref{theorem:coverageprobabilitywithblockage}. We compute the coverage probability for the typical user in a mmWave CoMP heterogenous network in the presence of blockage. The examples provided illustrate scenarios when cooperation is beneficial (in terms of increasing the coverage probability) and examples when the increase is not substantial. Numerical results suggest that the former is in fact the case when the mmWave network is dense (captured by the tier radius and consequently its intensity) - a feature expected for millimeter wave networks~\cite{survey:whatwill5gbe},~\cite{taimmWavedenseconference}. This can be interpreted as follows, with extreme densification, the number of LOS interfering base stations increases and thus interference increases, a remark also noted in~\cite{taimmWavedenseconference}. Therefore, cooperation limits the interference by increasing the number of serving base stations and therefore providing higher coverage probabilities.

{\bf Example 2:}
In Fig.~\ref{fig:blockageexample1} we plot the coverage probabilities for a tier with parameters given in Table~\ref{table:TableNumericalExamplewithblockage1}. A tier with average radius of $80$ meters is considered with blockage parameters $\beta=0.003, 0.006, 0.0143$ (corresponding to average LOS range which is greater than 80 m for $\beta=0.003$ and an average range that cannot reach a user at the edge for $\beta$=0.0143). 
The coverage probability for this one tier CoMP mmWave network with $n=2$ and with base station power available $P=1$W is compared with the following cases, Case 1) a one tier mmWave network with no base station cooperation $(n=1)$, with a base station transmit power $P=1$W, Case 2) a one tier mmWave network with no base station cooperation $n=1$, but with base station transmit power equal to the sum of transmit power if two base stations were to cooperate $P=2$W. 
The increase in coverage probability for both cases is approximately an increase of 0.12 in probability for a threshold $T=5,10$ dB. Moreover, it is interesting to note that an increase in the tier blockage parameter (shorter range of LOS links) would increase the coverage probability. This can be interpreted as follows: an increase in the blockage parameter increases the probability of blockage of the interfering LOS base stations, resulting in higher coverage probabilities. The curves corresponding to the coverage probabilities to Case 1 and Case 2 are very close since the power at the interfering base stations has also increased with this assumption (which also means that this network is not noise limited). We shall show in the subsequent example, that this observation doesn't hold for a less dense tier with a high probability of NLOS base stations.
\begin{figure}
\centering
\includegraphics[width=10cm]{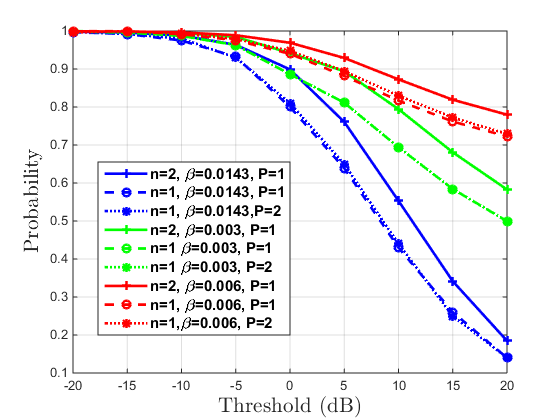}
\caption{Coverage probability in Th.~\ref{theorem:coverageprobabilitywithblockage} for a one-tier network with parameters in Table~\ref{table:TableNumericalExamplewithblockage1} with two cooperating base stations $(n=2)$ and without base station cooperation $(n=1)$ and for different blockage parameters.}
\label{fig:blockageexample1}
\end{figure}

{\bf Example 3:}
A tier with an average radius of 250 meters and with tier parameters given in Table~\ref{table:TableNumericalExamplewithblockage2} and with a blockage parameter $\beta=0.02$ (corresponding to a high probability of blockage and average LOS range of 50 m) is plotted in Fig.~\ref{fig:blockageexample2}. The increase in coverage probability due to cooperation in this case is minimal and is approximately 0.05 at all thresholds. This can be interpreted as follows: 1) a tier with high blockages will also block interfering signals and 2) when the density of base stations is not too dense, the $n$ strongest base stations are not {\it too strong} to cause a substantial increase in coverage probability due to the fact that distance at which these cooperating base stations are located increases too (thus received power decreases). As seen in Fig.~\ref{fig:blockageexample2}, increasing the power at the base station (but no cooperation) provides higher coverage probability than the case with base station cooperation. We shall show that the observations made for this example do not hold when there is no fading on the direct links from the cooperating base stations in Section~\ref{sec:coveragenofading}.
\begin{figure}
\centering
\includegraphics[width=10cm]{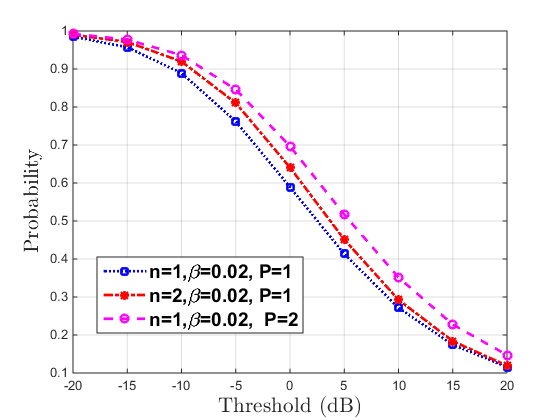}
\caption{Coverage probability in Th.~\ref{theorem:coverageprobabilitywithblockage} for a one-tier network with parameters in Table~\ref{table:TableNumericalExamplewithblockage2} with two cooperating base stations $(n=2)$ and without base station cooperation $(n=1)$.}
\label{fig:blockageexample2}
\end{figure}

{\bf Example 4:}
{ In an attempt to understand whether the observations hold for a network which is not noise limited as in Example 2 but for a larger number of antennas at the base stations, we consider the example of a dense mmWave tier network (average radius of 50 m) and $N_t=64$ and with network parameters as in Table~\ref{table:TableNumericalExamplewithblockage3}. In Fig.~\ref{fig:blockageexample3} we plot the coverage probabilities corresponding for the different cases which are described in Example 1. The observations made in Example 1 hold for this example too with almost the same increase (11\%) in coverage probability for a threshold $T=10,15$ dB.  }

\begin{figure}
\centering
\includegraphics[width=10cm]{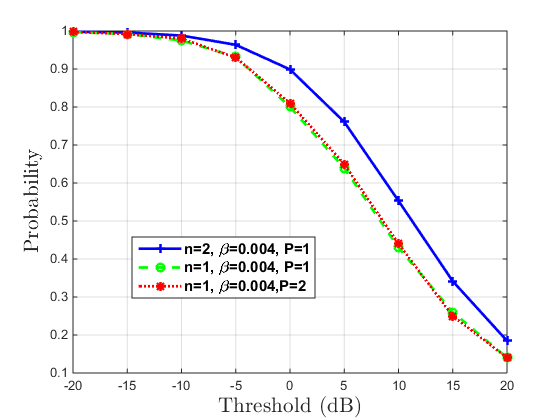}
\caption{Coverage probability in Th.~\ref{theorem:coverageprobabilitywithblockage} for a one-tier network with parameters in Table~\ref{table:TableNumericalExamplewithblockage3} with two cooperating base stations $(n=2)$ and without base station cooperation $(n=1)$.}
\label{fig:blockageexample3}
\end{figure}

\begin{table*}
\parbox{.45\linewidth}{
\centering
\caption{Tier parameters for Fig.~\ref{fig:blockageexample2} (Example 3)} 
\label{table:TableNumericalExamplewithblockage2}
\begin{tabular}{| l | l |}
\hline
Parameter  & Value  \\ 
\hline 
Intensity& $\lambda= (250^2\pi)^{-1}$\\
\hline
Power &$P_1=1$ W \\
\hline 
Path Loss& $\alpha_1=2, \alpha_2=4$ \\
\hline 
 Antennas & $N_t =64$\\
\hline 
Noise Figure & 5 dB\\
\hline
Blockage & $\beta=0.02$\\
\hline
Bandwidth & 1 GHz\\
\hline
\end{tabular}
}
\hfill
\parbox{.45\linewidth}{
\centering
\caption{Tier parameters for Fig.~\ref{fig:blockageexample3} (Example 4)} 
\label{table:TableNumericalExamplewithblockage3}
\begin{tabular}{| l | l |}
\hline
Parameter  & Value  \\ 
\hline 
Intensity& $\lambda= (50^2\pi)^{-1}$\\
\hline
Power &$P_1=1$ W \\
\hline 
Path Loss& $\alpha_1=2, \alpha_2=4$ \\
\hline 
 Antennas & $N_t =64$\\
\hline 
Noise Figure & 5 dB\\
\hline
Blockage & $\beta=0.004$\\
\hline
Bandwidth & 1 GHz\\
\hline
\end{tabular}}
\end{table*}


\section{Coverage Probability with Nakagami fading and blockage}
\label{sec:Coverage Probability with Nakagami fading and blockage}
In this section we consider the same network model as in Section~\ref{sec:network model blockage} but choose a different fading distribution on the channel gains from the strongest cooperating base stations, similar to~\cite{salammmWavecoverage}. In particular we consider Nakagami fading with parameter $m$, while keeping the same assumption of Rayleigh fading for the interfering channel gains. 

Using the coverage probability expression for a general fading distribution~\cite[{Eq. 2.11}]{baccellialoha}, we are then able to derive an upper bound on the coverage probability for this network. We then consider another upper bound by evaluating the network in the absence of interference. The coverage probability is then defined as the probability that the signal-to-noise-ratio (SNR) is greater than a certain threshold. Using complex analysis methods of integration we derive closed form for the coverage probability.

\subsection{Performance Analysis}
\label{sec:performance analysis with blockage and Nakagami fading}
\begin{thm}\label{theorem:coverageprobabilitywithblockageNakagamifading}
An upper bound on the coverage probability for the typical user, in a downlink mmWave heterogenous network with blockage with $K$ tiers, where each tier has a blockage parameter $\beta_k$, and with $n$ base stations having ULA with $N_t$ antennas jointly transmitting to it, with the assumption of Nakagami fading on the cooperating channel gains is
\begin{align} \label{eq:coverageprobabilitywithblockageNakagamifading}
\normalfont \mathbb{P}({\text{SINR} > T})\leq\int\limits_{0<\gamma'_1<\cdots<\gamma'_n<+\infty}f_{\Gamma'}(\gamma')\int_{-\infty}^{\infty}\mathcal{L}_I( 2j\pi T'  s)\mathcal{L}_{N}(2j \pi T' s)\frac{\mathcal{L}_{S}^{\rm{UP}}( - 2 j\pi   s)-1}{2j\pi s} ds \ d\gamma'
\end{align}
where  $\gamma_{i}'=\frac{\| v_i\| ^\alpha}{P_{f(v_i)}}$ for  $i \in [1:n]$ and $T'= \frac{T}{\sum_{i\leq n} \gamma_i'^{-1}} $ while the joint distribution of $\gamma'=[\gamma'_{1},\cdots,\gamma'_{n}]$ is given by~\eqref{eq:distributionofnbasestations} and where the Laplace transform of $I$ (assuming Rayleigh fading on the interfering links) is given by~\eqref{eq:Laplaceinterferencecase1} while the intensity and intensity measure are given by~\eqref{eq:intensitywithblockage} and~\eqref{eq:intensitymeasurewithblockage} respectively,
and where
\begin{align}
 &\mathcal{L}_S^{\rm{UP}}(s)= \frac{1}{(1+s/m)^{nm}}\\
 &\mathcal{L}_{N}(s) =e^{-s\sigma^2/N_t}.
\end{align}
\end{thm}
\begin{proof}
Please refer to Appendix~\ref{appendix:Nakagamifadingcase} for a detailed proof. 
\end{proof}
\begin{rem}\label{rem:notanupperbound}
For the case of no base station cooperation $n=1$, the coverage probability in~\eqref{eq:coverageprobabilitywithblockageNakagamifading} above is exact and is not an upper bound.
\end{rem}
\begin{rem}
If one desires the exact coverage probability, one can obtain it with the theorem below, with ${\cal L}_S^{\rm UP}(s)$ replaced by the true Laplace transform of the signal $S=\big|\sum\limits_{i=1}^n\sqrt{\gamma_{v_i}} h_{v_i}|^2$. 
\end{rem}

\begin{cor}
\label{corr: SNRnakagmi}
An upper bound on the coverage probability in the absence of interference for the typical user, in a downlink mmWave heterogenous network with blockage with $K$ tiers, where each tier has a blockage parameter $\beta_k$, and with $n$ base stations having ULA with $N_t$ antennas jointly transmitting to it, with the assumption of Nakagami fading on the cooperating channel gains is 
\begin{align}
\normalfont \mathbb{P}({\text{SNR} > T}) \leq \int\limits_{0<\gamma'_1<\cdots<\gamma'_n<+\infty} f_{\Gamma'}(\gamma') \, \left(\frac{g^{(nm-1)}(z^\ast)}{(nm-1)!}\right) \, d\gamma'. \label{eq:coverageprobabilitySNR}
\end{align}
 where  $z^{\ast}=\frac{m}{2\pi j}$, $\gamma_{i}'=\frac{\| v_i\| ^\alpha}{P_{f(v_i)}}$ for  $i \in [1:n]$ and $T'= \frac{T}{\sum_{i\leq n} \gamma_i'^{-1}} $ while the joint distribution of $\gamma'=[\gamma'_{1},\cdots,\gamma'_{n}]$ is given by~\eqref{eq:distributionofnbasestations} with an intensity and intensity measure as in~\eqref{eq:intensitywithblockage} and~\eqref{eq:intensitymeasurewithblockage} respectively and where
\begin{align}
g(z)=(-1)^{nm}\frac{1-(1-\frac{2\pi jz}{m})^{nm}}{(\frac{2\pi j}{m})^{nm}(2\pi jz)}e^{-2\pi jz  \frac{T'\sigma^2}{N_t}}
\end{align}
and where $g^{nm-1}(z)$ is the $(nm-1)$ derivative of the function $g(z)$.
\end{cor}

\begin{proof}
Please refer to Appendix~\ref{appendix:Nakagamifadingcase} for the proof.
\end{proof}

A similar remark to that made in Remark 4 can be made for the coverage probability (in absence of interference) in~\eqref{eq:coverageprobabilitySNR}.

\begin{table*}
\parbox{.45\linewidth}{
\centering
\centering
\caption{Tier Parameters for Fig.\ref{fig:nakagamiversusrayleigh} (Example 5)} 
\label{table:TableNumericalExample3}
\begin{tabular}{| l | l |}
\hline
Parameter  & Value  \\ 
\hline 
Intensity& $\lambda_1= (200^2\pi)^{-1}$\\
\hline
Power &$P_1=1$ W (30 dBm)\\
\hline 
Path Loss& $\alpha_1=2$ and $\alpha_2=4$ \\
\hline 
 Antennas & $N_t =16$\\
\hline 
Noise Figure & 5 dB\\
\hline
Bandwidth & 1 GHz\\
\hline
Blockage & 0.025\\
\hline
\end{tabular}
}
\hfill
\parbox{.45\linewidth}{
\centering
\caption{Tier parameters for Fig.~\ref{fig:nofadingexample1} (Example 6)} 
\label{table:TableNumericalExamplenofading1}
\begin{tabular}{| l | l |}
\hline
Parameter  & Value  \\ 
\hline 
Intensity& $\lambda= (200^2\pi)^{-1}$\\
\hline
Power &$P_1=1$ W \\
\hline 
Path Loss& $\alpha_1=2, \alpha_2=4$ \\
\hline 
 Antennas & $N_t =64$\\
\hline 
Noise Figure & 5 dB\\
\hline
Blockage & $\beta=0.025$\\
\hline
Bandwidth & 1 GHz\\
\hline
\end{tabular}}
\end{table*}

\subsection{Numerical Results}
{\bf Example 5:}
We consider a one tier network with two cooperating base stations $n=2$ and with tier parameters given in Table~\ref{table:TableNumericalExample3}. The coverage probability in~\eqref{eq:coverageprobabilitywithblockageNakagamifading} for the case of $m=3$ with and without base station cooperation are plotted. The purpose of this numerical example is to show that for tiers with high probability of blockage, in this case taken to be $\beta=0.025$ (corresponding to a high probability of blockage and average LOS range of 40 m), evaluation of the coverage probability of the network with and with the absence of interference yields almost exact numerical results. Therefore, we fix $n=2$ and we plot the coverage probability in~\eqref{eq:coverageprobabilitywithblockageNakagamifading} and~\eqref{eq:coverageprobabilitySNR} for Case 1) $m=3$ and $n=2$. While we use Th.~\ref{theorem:coverageprobabilitywithblockage} to plot the coverage probability for the Rayleigh fading Case 2) $m=1$ and $n=2$ (also with and without interference). The two curves shown in Fig.~\ref{fig:nakagamiversusrayleigh} for each of the cases corresponding to the Rayleigh and Nakagami fading almost exactly overlap. 
\begin{figure}
\centering
\includegraphics[width=10cm]{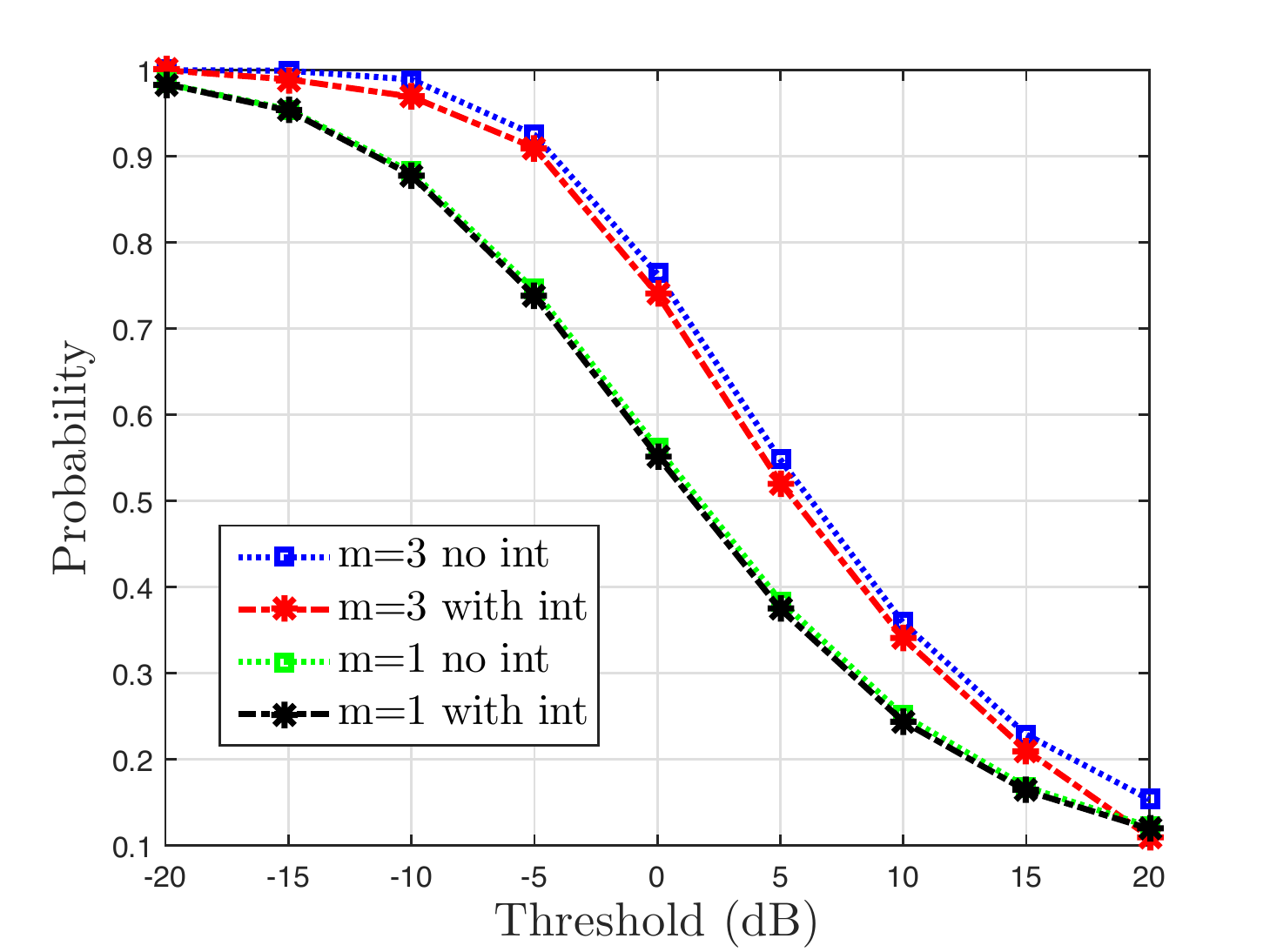}
\caption{Upperbounds on coverage probability in Th.~\ref{theorem:coverageprobabilitywithblockageNakagamifading} and Corollary~\ref{corr: SNRnakagmi} for a one tier network with parameters in Table~\ref{table:TableNumericalExample3} and for $n=2$ (two cooperating base stations) and for $n=1$ (no cooperation). The lower bounds are plotted using Th.~\ref{theorem:coverageprobabilitywithblockage}}. 
\label{fig:nakagamiversusrayleigh}
\end{figure}


\section{Coverage Probability with No Small Scale Fading}
\label{sec:coveragenofading}
{ In this section we consider the same network model as in Section~\ref{sec:network model blockage} but where the cooperating channel gains do not experience any fading. As shown in~\cite{Rappa_itwillwork
}, the assumption of having no small scale fading from the serving base stations is a good assumption in mmWave systems due to the highly directional transmission and when the receivers are not present in a rich scattering environment (in rich scattering environments Rayleigh fading may be more reasonable). The Rayleigh distribution is used to model the fading distribution of the interfering channel gains. Interestingly, we will show through numerical examples that the increase in coverage probability with two cooperating base stations is more pronounced when the cooperating channels experience no fading than that obtained when the fading is assumed to be the Rayleigh fading. }

\label{sec:CoverageProbabilitywithNoSmallScaleFading}
\begin{thm}
\label{thm:nofadingcaseSNRcase}
The coverage probability in the absence of small scale fading for the typical user, in a downlink mmWave heterogenous network with blockage with $K$ tiers, where each tier has a blockage parameter $\beta_k$, and with $n$ base stations having ULA with $N_t$ antennas jointly transmitting to it, with the assumption of no fading on the cooperating channel gains is 
\begin{align} \label{eq:coverageprobabilitywithblockagenosmallscalefading}
\normalfont \mathbb{P}({\text{SINR} > T})=\int\limits_{0<\gamma'_1<\cdots<\gamma'_n<+\infty}f_{\Gamma'}(\gamma')\int_{-\infty}^{\infty}\mathcal{L}_I( 2j\pi T  s)\mathcal{L}_{N}(2j \pi T s)\frac{\mathcal{L}_{S}( - 2 j\pi   s)-1}{2j\pi s} ds \ d\gamma'
\end{align}
where  $\gamma_{i}'=\frac{\| v_i\| ^\alpha}{P_{f(v_i)}}$ for  $i \in [1:n]$ while the joint distribution of $\gamma'=[\gamma'_{1},\cdots,\gamma'_{n}]$ is given by~\eqref{eq:distributionofnbasestations}
with an intensity and intensity measure as in~\eqref{eq:intensitywithblockage} and~\eqref{eq:intensitymeasurewithblockage} respectively and where
$\mathcal{L}_S(s)= e^{-s(\sum_{i\leq n}\gamma_i'^{-1/2})^2}$ and $\mathcal{L}_{N}(s) =e^{-s\sigma^2/N_t}$.
\end{thm}
\begin{proof}
The proof follows easily from finding the Laplace transform of the desired signal.
\end{proof}
\subsection{Numerical Results}

\begin{figure*}
\centering
\subfigure[Coverage probability in Th.~\ref{thm:nofadingcaseSNRcase} for a one tier network with parameters in Table~\ref{table:TableNumericalExamplenofading1} with two cooperating base stations ($n=2$) and with no cooperation ($n=1$). The coverage probabilities for the same network parameters but with Rayleigh assumption in Th.~\ref{theorem:coverageprobabilitywithblockage} are also plotted for ($n=1$) and ($n=2$). ]{%
\includegraphics[width=0.45\textwidth]{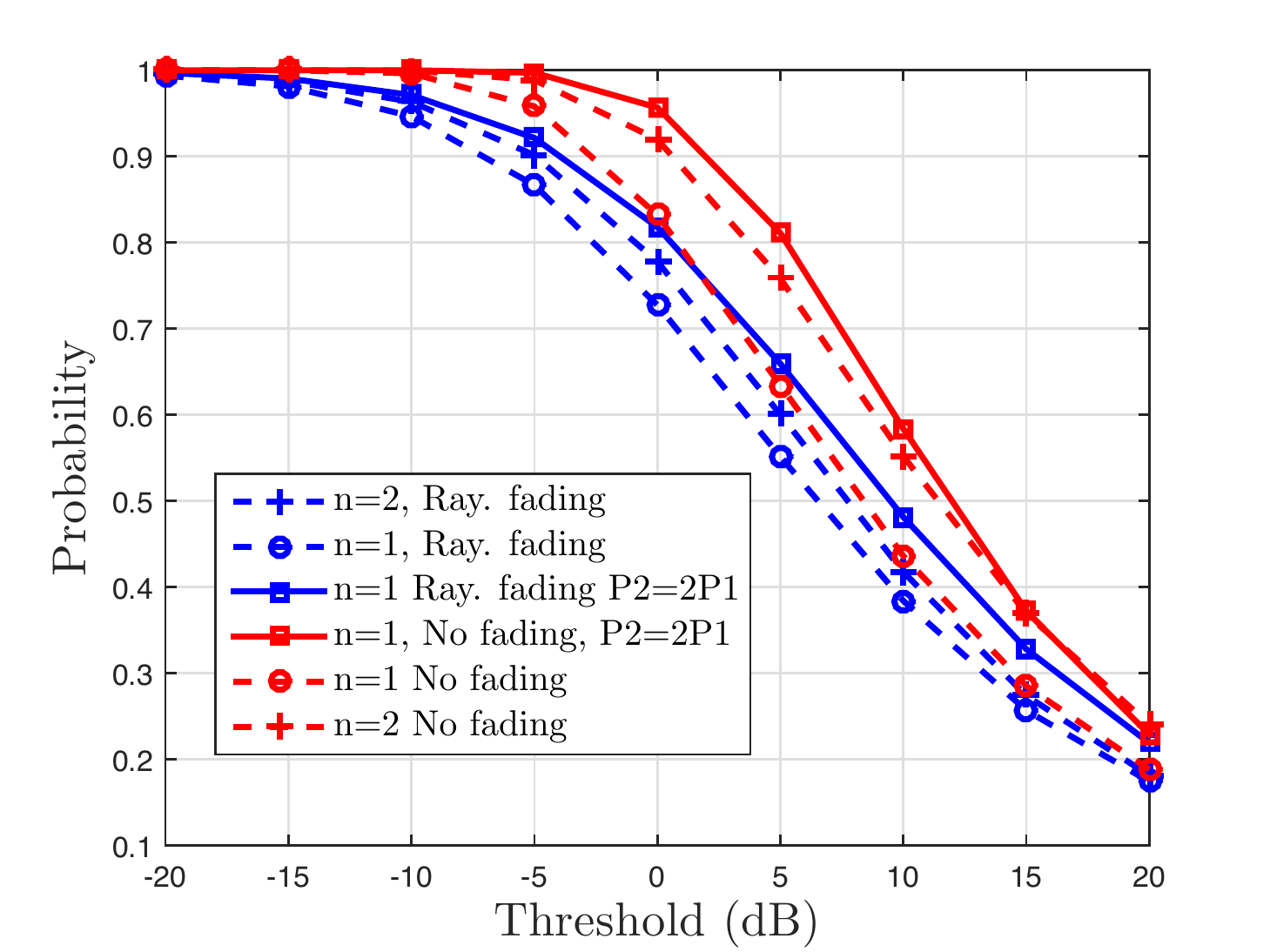}%
\label{fig:nofadingexample1}%
}
\hfil
\subfigure
[Coverage probability in Th.~\ref{thm:nofadingcaseSNRcase} for a one tier network with parameters in Table~\ref{table:TableNumericalExamplewithblockage2} for $\beta=0.003$ with two cooperating base stations ($n=2$) and with no cooperation ($n=1$). The coverage probabilities for the same network parameters but with Rayleigh assumption in Th.~\ref{theorem:coverageprobabilitywithblockage} are also plotted for ($n=1$) and ($n=2$).]{%
\includegraphics[width=0.45\textwidth]{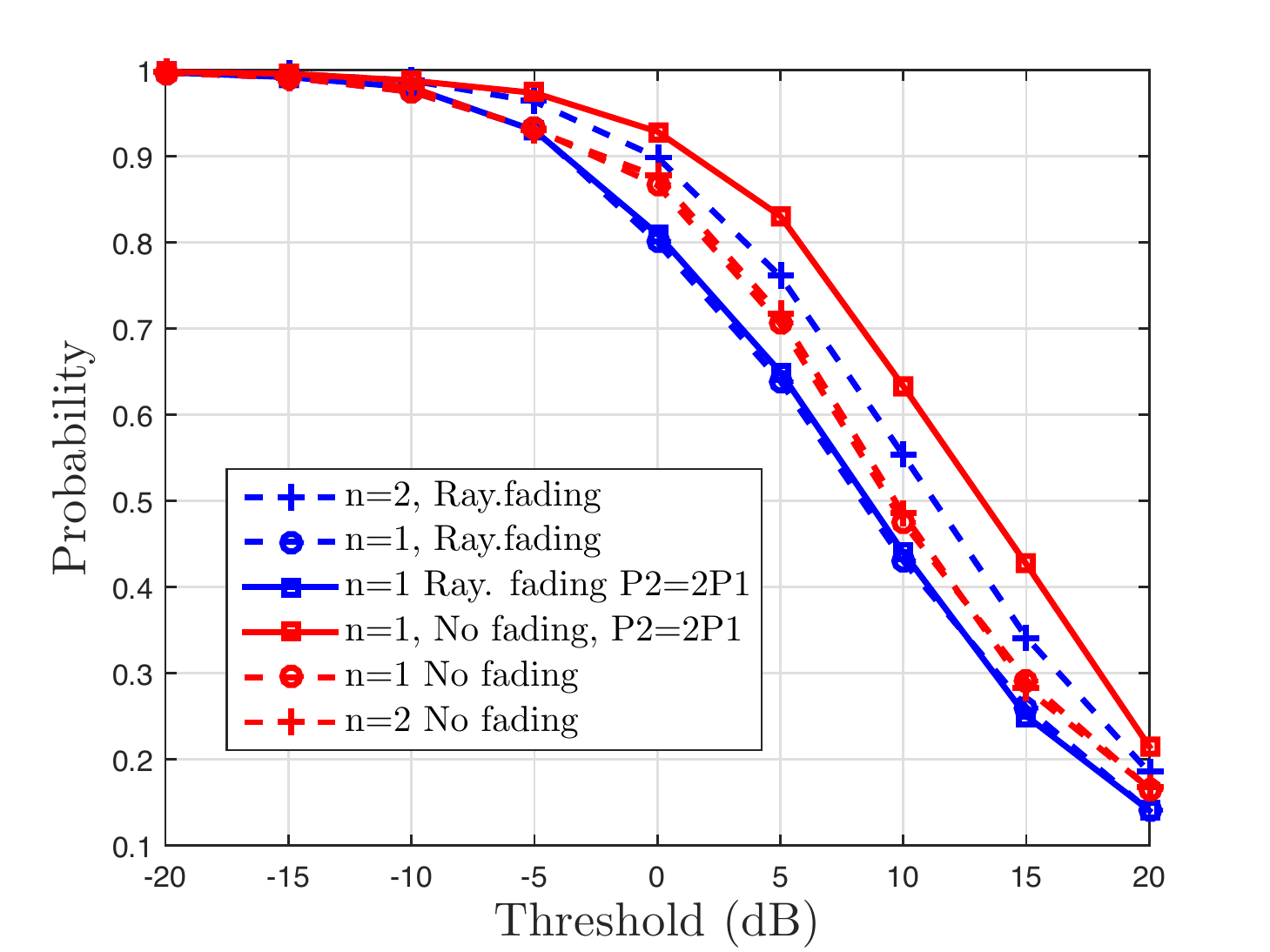}%
\label{fig:nofadingexample2}%
}
\end{figure*}
{\bf Example 6:}
{ In this example we numerically evaluate Th.~\ref{thm:nofadingcaseSNRcase}. We consider a one tier network with tier parameters given in Table~\ref{table:TableNumericalExamplenofading1}. The coverage probability in Th.~\ref{thm:nofadingcaseSNRcase} with two cooperating base stations $n=2$ is plotted against the two different cases described as in Example 2 of Section~\ref{sec:performance analysis with blockage}. We also plot the three curves corresponding to the Rayleigh fading assumption on the cooperating channel gains in Th~\ref{theorem:coverageprobabilitywithblockage}. The curves in Fig.~\ref{fig:nofadingexample1} suggest that the increase in coverage probability with CoMP is mainly due to a power increase since the curve corresponding to the case when there is no cooperation but with double transmit power is very close to that of the curve of having two cooperating base stations but with half the transmit power. This also implies that this network is noise limited. This observation is valid for the networks where gains from cooperating base stations are non fading or when they are Rayleigh fading. It is interesting to note that the increase in probability with two cooperating base stations (almost 18 \% for threshold $T=5$ dB) for the case when there is no small scale fading exceeds the increase with CoMP (5 \% for the same threshold) for the case when the channel gains are Rayleigh distributed and hence implying that CoMP provides larger gains for mmWave networks with no small scale fading from the serving base stations.

{\bf Example 7:} 
In an attempt to understand whether the observations made in Example 6 hold for a network which is not noise limited, we consider the network in Example 2 ($\beta$=0.003) and compare the Rayleigh fading case to the case when there is no small scale fading. In Fig.~\ref{fig:nofadingexample2} we plot the curves corresponding to the different cases as explained in Example 2 along with the curves corresponding to the coverage probabilities non fading networks in Th.~\ref{thm:nofadingcaseSNRcase}. The increase in coverage probability with CoMP ($n=2$) is an increase of 16\% at threshold $T=10$, while the increase with CoMP for the Rayleigh fading case is almost 12 \% at the same threshold. Therefore, the observation made in Example 6 does in fact hold for this network too.}

\section{Conclusions}
In this paper we have considered the problem of base station cooperation in mmWave heterogenous networks. Using stochastic geometry, coverage probabilities were derived at the typical user, accounting for directionality at the base stations, blockage, interference and different fading distributions (Rayleigh and Nakagami). Numerical results suggest that coverage with cooperation rival that with no cooperation especially in dense mmWave networks with no small scale fading on the cooperating channel gains. Future work includes deriving the coverage probability at a multi-antenna typical receiver and accounting for possible errors due to beam steering.


\appendices
\section{Proof of Theorem~\ref{theorem:coverageprobabilityULAanalysis}}
\label{appendix:coverageprobabilityULAanalysis}
The analysis of the coverage probability for the mmWave heterogeneous network is similar to that in~\cite[Appendix A]{gauravmicrowavecomp} with two major differences. The first difference is the presence of multiple antennas at the transmitter. The second difference is that the interference is a function of i.i.d uniformly distributed random variables, assuming that the path angles are independent and uniformly distributed over $[-\pi,+\pi]$. 

 Let $\Theta_k = \{ \frac{{\| v\| ^\alpha}}{{P_k}}, v \in \Phi_{k}\}$ for $k\in[1:K]$. Its density can be derived using the Mapping theorem~\cite[Thm. 2.34]{book:martinhaenggi} and is given by
 \begin{align}\label{eq:intensity of the mapped process}
 \lambda_k(v)= \lambda_k\frac{2\pi}{\alpha}P_{k}^{\frac{2}{\alpha}}v^{\frac{2}{\alpha}-1}, \,\, k\in[1:K].
 \end{align} The process $\Theta=\cup_{k=1}^K \Theta_k$ has a density 
\begin{align}\lambda(v)=\sum_{k=1}^K\lambda_k(v)\label{eq:intensityofunionofthetiers}.\end{align}


\subsection{ Distribution of Strongest Base stations}
\vspace{0.4cm}
We assume that the elements in the process $\Theta$ are indexed in increasing order. Let $$\gamma'_{i} = \frac{\| v_i\| ^\alpha}{P_{f(v_i)}}$$ then ${\bf \gamma'}=\{\gamma'_{1},\cdots,\gamma'_{n}\} $ denotes the set of {\it normalized pathloss} of the cooperating base stations. We first present the distribution of the two nearest base stations by following similar steps as done in~\cite{Motchanov_Distance}, then derive the distribution of $n$ closest base stations. The distribution of the closest two base stations (assuming two cooperating base stations) is given by 
\begin{align} \label{eq:joint distribution}
f_{\Gamma'} (\gamma'_1,\gamma'_2)= f_{\Gamma_2'|\Gamma_1'}(\gamma'_2|\gamma'_1) f_{\Gamma_1'}(\gamma'_1)
\end{align}
where the distribution of the first base station with strong received power is obtained from the null probability of a PPP is 
\begin{align} \label{eq:firstbasestationsestationprob}
f_{\Gamma_1'}( \gamma_1') =&\lambda(\gamma_1')e^{-\Lambda(\gamma_1')}
\end{align}
while the conditional distribution is given by
\begin{align}\label{eq:conditionalprob}
f_{\Gamma_2'|\Gamma_1'}(\gamma_2'| \gamma_1')=&\lambda(\gamma_2')e^{-\Lambda(\gamma_2')+\Lambda(\gamma_1')}
\end{align}
The joint distribution for the case of $n=2$ base stations is obtained by substituting~\eqref{eq:firstbasestationsestationprob} and~\eqref{eq:conditionalprob} in~\eqref{eq:joint distribution}.
The result can be generalized to any number $n$ of cooperating base stations
\begin{align} 
f_{\Gamma'} (\gamma')=&\prod_{i=1}^n \lambda(\gamma_i')e^{-\Lambda(\gamma_n')}.
\end{align}

\subsection{Derivation of Coverage Probability}

We have assumed that the cooperating base stations have {\it normalized pathloss} $\gamma'_i$ by $i \leq n$, then the desired signal power at the numerator of~\eqref{eq:SINRexpressionomnicase} can be re-written as
\begin{align}\nonumber
S=&\big|\sum\limits_{i=1}^n\sqrt{\gamma_{v_i}} h_{v_i}|^2 =\big|\sum\limits_{i\leq n} \gamma_{i}'^{-1/2} h_{i} \big|^2.
\end{align}

{We have that the interfering base stations are indexed with $i>n$, then the power of the interference $I$ can be expressed (by replacing the index $l_i$ with just $i$) as
\begin{align}\nonumber
I &=\sum\limits_{i=1}^{|\mathcal{T}|^c}\ \gamma_{l_i} |h_{l_i}|^2 |G_t(\Omega_{\phi_{l_i}^t}-\Omega_{\theta_{l_i}^t})|^2\\
&=\sum\limits_{i>n} \gamma_{i}'^{-1} |h_{i}|^2\bigl |G_t(\Upsilon_i)\bigl |^2\label{eq:interferenceexpression}
\end{align}

The coverage probability for a threshold $T$, can be re-written as
\begin{align*}\nonumber
\mathbb{P}(\text{SINR} > T)&= \mathbb{P}\left (S> T(I+\frac{\sigma^2}{N_t})\right)= \mathbb{E}_{ \gamma',I} \Bigg[\mathbb{P}\left (\big|\sum\limits_{i\leq n} \gamma_{i}'^{-1/2} h_{i} \big|^2> T(I+\frac{\sigma^2}{N_t})\Biggl|\gamma',I\right)\Bigg] \\\nonumber
&\stackrel{(a)}{=} \mathbb{E}_{ \gamma',I} \Bigg[\exp\left(\frac{-T(I+\frac{\sigma^2}{N_t})}{\sum_{i\leq n} \gamma_{i}'^{-1}}\right)\Bigg]\nonumber\stackrel{(b)}{=}\mathbb{E}_{\gamma'} \Bigg[\mathcal{L}_I\left(\frac{T}{ \sum_{i\leq n}\gamma_{i}'^{-1}}\right)\mathcal{L}_N\left(\frac{T}{ \sum_{i\leq n}\gamma_{i}'^{-1}}\right)\Bigg]\\
&\stackrel{(c)}{=}\int\limits_{0<\gamma'_1<\cdots<\gamma'_n<+\infty} \mathcal{L}_I\bigg(\frac{T}{ \sum_{i\leq n}\gamma_{i}'^{-1}}\bigg)\mathcal{L}_N\left(\frac{T}{ \sum_{i\leq n}\gamma_{i}'^{-1}}\right)d\gamma'
\end{align*}
where (a) follows from the cumulative density function of the exponentially distributed random variable $S$ (due to Rayleigh fading assumption) with mean $\sum\limits_{i\leq n} \gamma_{i}'^{-1}$;
(b) follows from the definition of the Laplace transform of $I$, $\mathcal{L}_I(s)= \mathbb{E} [ e^{-sI}]$ and the Laplace transform of the noise, $\mathcal{L}_N(s)=\mathbb{E} [ e^{-s\sigma^2/N_t}]$;
(c) by definition of the expectation with respect to the distribution of $\gamma'$.

\vspace{0.4cm}

Next we evaluate the Laplace transform of the interference $I$, but before going into the details of the derivation we need to find the distribution of 
$\Upsilon_i:=\Omega_{\phi_{l_i}^t}-\Omega_{\theta_{l_i}^t}$, since the interference in~\eqref{eq:interferenceexpression} is a function of the beam forming gain function which in turn is a function of $\Upsilon_i$. The beam forming gain is given by~\eqref{eq:antennagainfunction}.

With the assumption that the interfering path angles of departure and the beam steering angle used by the interfering base stations are i.i.d $\sim U([-\pi,+\pi])$, then the directional cosine $ \Omega_{\phi_{l_i}^t}$ and $\Omega_{\theta_{l_i}^t}$ are random variables with the following common probability density function

\[ f_\Omega(\omega) = \left\{ \begin{array}{ll}
         \frac{1}{\pi\sqrt{1-\omega^2}} & \mbox{if $-1\leq \omega \leq 1$};\\
          0 & \mbox{otherwise}.\end{array} \right. \]

then the distribution of $\Upsilon_i=\Omega_{\phi_{l_i}^t}-\Omega_{\theta_{l_i}^t}$ is the result of the convolution of the probability density functions of $\Omega_{\phi_{l_i}^t}$ and $\Omega_{\theta_{l_i}^t}$ and is given by
\begin{align}\label{eq:pdfofthedifferencebetweendirectionalcosines}
f_{\Upsilon_i}(\varepsilon_i)= \int_{\max\{-1,-1-\varepsilon_i\}}^{\min\{1,1-\varepsilon_i\}} \left( \frac{1}{\pi^2\sqrt{1-(\varepsilon_i+\omega)^2}} \frac{1}{\sqrt{1-\omega^2}}\right)\, dy
\end{align}

Then the Laplace transform of the interference can be derived
\begin{align} \nonumber
\mathcal{L}_I(s)& =\mathbb{E}\Bigg[e^{-s\sum\limits_{i>n} \gamma_{i}'^{-1} |h_{i}|^2 |G_t(\Upsilon_i)|^2}\Bigg]=\mathbb{E}\Bigg[\prod_{i>n}\bigg(e^{-s  \gamma_{i}'^{-1}|h_{i}|^2|G_t(\Upsilon_i)|^2 }\bigg)\Bigg]\\\nonumber
&\stackrel{(a)}{=}\mathbb{E}_{\{\Upsilon_i\},\Theta}\Bigg[\prod_{i>n}\mathbb{E}_{|h|^2}\bigg(e^{-s \gamma_{i}'^{-1}|h|^2|G_t(\Upsilon_i)|^2 }\bigg)\Bigg]\\\nonumber
&\stackrel{(b)}{=}\mathbb{E}_{\{\Upsilon_i\},\Theta }\Bigg[\prod_{{i>n}}\Bigg(\frac{1}{1+s |G_t(\Upsilon_i)|^2 \gamma_{i}'^{-1} }\Biggl)\Bigg]\\\nonumber
&\stackrel{(c)}{=}\mathbb{E}_{\Theta }\Bigg[\prod_{{i>n}}\mathbb{E}_{\Upsilon}\Bigg(\frac{1}{1+s |G_t(\Upsilon)|^2 \gamma_{i}'^{-1} }\biggl)\Bigg]\\\nonumber
&\stackrel{(d)}{=}\mathbb{E}_{\Theta }\Bigg[\prod_{{i>n}}\left(\int_{-2}^{+2}\Bigg(\frac{1}{1+s|G_t(\varepsilon)|^2 \gamma_{i}'^{-1} }\Biggl) f_{\Upsilon}(\varepsilon)\ d\varepsilon \ \right)\Bigg] \\\label{eq:Laplacetransformofinterference}
&\stackrel{(e)}{=}\text{exp} \Bigg( -\int_{\gamma_n'}^{\infty} \Bigg[ 1- \int_{-2}^{+2}\Bigg(\frac{1}{1+s|G_t(\varepsilon)|^2 v^{-1} }\biggl) f_{\Upsilon}(\varepsilon)\ d\varepsilon \Bigg] \ \lambda(v) \, dv \Bigg)
\end{align}
where (a) follows from the i.i.d distribution of $|h_i|^2$ and their independence from $\Theta$ and $\Upsilon_i$;
where (b) follows from the Rayleigh fading assumption and the moment generating function of an exponential random variable;
where (c) follows from the i.i.d distribution of $\Upsilon_i$ and their independence from $\Theta$;
(d) from the taking the expectation with respect to the random variable $\Upsilon_i$ whose distribution is given by~\eqref{eq:pdfofthedifferencebetweendirectionalcosines};
(e) follows from the probability generating function of poisson point process~\cite[Thm. 4.9]{book:martinhaenggi} (as used in~\cite[Eq. (38)]{Nigam_CoMP_journal}) and where $\lambda(v)$ is given by~\eqref{eq:intensity of the mapped process}.

Next we give an approximation of the Laplace transform of the interference for easier numerical evaluations (by approximating the beam forming gain function by a piecewise linear function) and to compare our results with~\cite{Nigam_CoMP_journal}, the Laplace transform of the interference is}
\begin{align*} 
\mathcal{L}_I(s)& 
\stackrel{(d)}{=}\mathbb{E}_{\Theta }\Bigg[\prod_{{i>n}}\left(\int_{-2}^{+2}\Bigg(\frac{1}{1+s |G_t(\varepsilon)|^2 \gamma_{i}'^{-1} }\Biggl) f_{\Upsilon}(\varepsilon)\ d\varepsilon \ \right)\Bigg]\\
&\stackrel{(e)}{\approx}\mathbb{E}_{\Theta }\Bigg[\prod_{{i>n}}\left(\int_{-2}^{-1/L_t} f_{\Upsilon}(\varepsilon)\ d\varepsilon +\int_{-1/L_t}^{1/L_t} \frac{1}{1+s \gamma_{i}'^{-1} }f_{\Upsilon}(\varepsilon)\ d\varepsilon +\int_{1/L_t}^{2} f_{\Upsilon}(\varepsilon)\ d\varepsilon \right)\Bigg]\\ 
&\stackrel{(f)}{=}\mathbb{E}_{\Theta }\Bigg[\prod_{{i>n}}\left(1 - \frac{c \ s\gamma_i^{-1}}{1+s\gamma_i^{-1}}\right)\Bigg]\stackrel{(g)}{=}\text{exp}\Bigg(-\int_{\gamma_n'}^{\infty}\Bigg[\frac{c \ sv^{-1}}{1+sv^{-1}}\Bigg]\lambda(v) \, dv\Bigg)
\end{align*}
where in (e) an approximation of the gain function was used which is given
$$G_t(\varepsilon) = \left\{ \begin{array}{ll}
         1 & \mbox{if $ - \frac{1}{L_t}\leq \varepsilon \leq \frac{1}{L_t} $}, \,\,\,\,\,\, L_t= N_t\Delta_t\\
         0 & \mbox{if $\mbox{otherwise}$}.\end{array} \right. $$       
(f) defining $c:=\int_{-1/L_t}^{+1/L_t} f_{\Upsilon}(\varepsilon)\ d\varepsilon$.

\begin{rem}
If $c=1$ then the Laplace transform in step (g) simplifies to that in~\cite[Eq. (38)]{Nigam_CoMP_journal}.
\end{rem}

\section{Proof of Theorem~\ref{theorem:coverageprobabilitywithblockage}}
\label{appendix:coveragewithblockage}

\subsection{Intensity and Intensity Measure}
Let $\Theta_k = \{ \frac{{\| v\| ^\alpha}}{{P_k}}, v \in \Phi_{k}\}$ for $k\in[1:K]$ with intensity $\lambda_k(v)$ given in~\eqref{eq:intensity of the mapped process}. The pathloss $\alpha$ is a random variable that takes on values $\alpha_1$ and $\alpha_2$ with probability $e^{-\beta_kv} $ and $1- e^{-\beta_kv}$ respectively (note that we have dropped the $\| .\| $ of $v$ for easier notation). Then the process $\Theta =\cup_{k=1}^K \Theta_k$ is a non-homogenous PPP with density $\lambda(v)= \sum_{k=1}^K \lambda_k(v)$.  In the following we compute the intensity and intensity measure of $\Theta_k$ for $k\in[1:K]$. By using the Mapping Theorem~\cite[Thm. 2.34]{book:martinhaenggi} the intensity measure and the intensity of each tier $k$, $k\in[1:K]$, are given by 
 \begin{align} \nonumber
\Lambda_k([0,r]) &= \int_{0}^{(rP_k)^{\frac{1}{\alpha_1}}} 2\pi \lambda_k v e^{-\beta_k v}dv + \int_{0}^{(rP_k)^{\frac{1}{\alpha_2}}} 2\pi \lambda_k v (1- e^{-\beta_k v})dv\\&=
\nonumber\frac{2\pi \lambda_k}{\beta_k^2}\left(1- e^{-\beta_k (rP_k)^{\frac{1}{\alpha_1}}}(1+\beta_k (rP_k)^{\frac{1}{\alpha_1}})\right)+ \pi \lambda_k (rP_k)^{\frac{1}{\alpha_2}}\\&-\frac{2\pi \lambda_k}{\beta_k^2}\left(1- e^{-\beta_k (rP_k)^{\frac{1}{\alpha_2}}}(1+\beta_k (rP_k)^{\frac{1}{\alpha_2}})\right)\\
&\lambda_k(v) =\frac{d\Lambda_k([0,v]) }{dv} =A_kv^{\frac{2}{\alpha_1}-1} e^{-a_kv^{\frac{1}{\alpha_1}}} +B_kv^{\delta_2-1}(1-e^{ -b_k v^{\frac{1}{\alpha_2}}})\label{eq:intensityforblockagecase}
\end{align}
with $A_k= \pi \lambda_k \frac{2}{\alpha_1} P_k^{\frac{2}{\alpha_1}}, a_k=\beta_kP_k^{\frac{1}{\alpha_1}}, b_k=\beta_kP_k^{\frac{1}{\alpha_2}} \text{and }B_k=\pi \lambda_k \frac{2}{\alpha_2}P_k^{\frac{2}{\alpha_2}}$.

\vspace{0.4cm}

The process $\Theta = \cup_{k=1}^K \Theta_k$ has the following intensity measure and intensity 
\begin{align} 
&\Lambda(v) = \sum_{k=1}^K\Lambda_k(v), \\
&\lambda(v) = \sum_{k=1}^K\lambda_k(v)= \sum_{k=1}^KA_kv^{\frac{2}{\alpha_1}-1} e^{-a_kv^{\frac{1}{\alpha_1}}} +B_kv^{\frac{2}{\alpha_2}-1}(1-e^{ -b_k v^{\frac{1}{\alpha_2}}}).\label{eq:intensitywithbockage}
\end{align}

\vspace{0.5cm}
\section{Proof of Theorem~\ref{theorem:coverageprobabilitywithblockageNakagamifading} and Corollary~\ref{corr: SNRnakagmi} }
\label{appendix:Nakagamifadingcase}
\subsection{Proof of Theorem~\ref{theorem:coverageprobabilitywithblockageNakagamifading}}
Let us re-consider a different distribution on the direct links - while keeping the same Rayleigh fading assumption on the interfering links -  in particular let us consider that the fading is Nakagami with shape parameter $m$ and scale parameter $\theta=1$.  In this case we will derive the distribution of an upper bound on the desired signal in particular the distribution of
\begin{align}\nonumber
S=& \big|\sum\limits_{i\leq n} \gamma_{i}'^{-1/2} h_{i} \big|^2 \leq \sum\limits_{i\leq n} \gamma_{i}'^{-1}\sum\limits_{i\leq n} |h_{i}|^2 =S^{\rm {UP}}
\end{align}
Then we have
$$\mathbb{P}\left(S^{\rm{UP}} \geq T( I+\frac{\sigma^2}{N_t})\right)=\mathbb{P}\left(\sum_{i\leq n}|h_{i}|^2\geq\frac{T(I+\frac{\sigma^2}{N_t})}{\sum_{i} \gamma_{i}^{-1}}\right)]=\mathbb{P}\left( S_{\rm {UP}}\geq T'(I+\frac{\sigma^2}{N_t})\right)\bigl|_{T'=\frac{T}{\sum_{i} \gamma_{i}^{-1}}, S_{\rm{UP}}=\sum_{i}|h_{i}|^2}$$ 
The distribution of the upper bound and its Laplace transform are respectively given by
$$ \sum_{i\leq n} |h_{i}|^2 \sim \text{Gamma}(nm,1/m) \leftrightarrow \mathcal{L}_{S^{\rm{UP}}}(s) =\frac{1}{(1+s/m)^{nm}}.$$
We then have from~\cite[{Eq. 2.11}]{baccellialoha} that the coverage probability is
\begin{align} \nonumber
\mathbb{P}\left(\text{SINR} > T\right)&= \mathbb{P}\left (S> T'(I+\frac{\sigma^2}{N_t})\right) \\\nonumber
&=\int\limits_{0<\gamma'_1<\cdots<\gamma'_n<+\infty}f_{\Gamma'}(\gamma')\int_{-\infty}^{\infty}\mathcal{L}_I( 2j\pi T'  s)\mathcal{L}_{N}(2j \pi T' s)\frac{\mathcal{L}_{S^{\rm UP}}( - 2 j\pi   s)-1}{2j\pi s} ds \ d\gamma' \label{eq:coveragenakagami}\nonumber
\end{align}
where the joint distribution of $\gamma'$ is given by
\begin{align}
f_{\Gamma'} (\gamma')=&\prod_{i=1}^n \lambda(\gamma_i')e^{-\Lambda(\gamma_n')}
\end{align}
Next we have from~\eqref{eq:Laplacetransformofinterference} the Laplace transform of the interference $\mathcal{L}_I(s)$ with an intensity $\lambda(v)$ given by~\eqref{eq:intensitywithbockage} while the Laplace transform of the noise is given by
\begin{align}
 &\mathcal{L}_{N}(s) = \mathbb{E}[ e^{-s\frac{\sigma^2}{N_t}}]=e^{-s\frac{\sigma^2}{N_t}}.
\end{align}


\subsection{Proof of Corollary~\ref{corr: SNRnakagmi}}
The coverage probability in the absence of interference is given by~\eqref{eq:coverageprobabilitywithblockageNakagamifading} with $\mathcal{L}_I(s)=1$ is
\begin{subequations}
\begin{align} 
\normalfont \mathbb{P}({\text{SNR} > T})&= \int\limits_{0<\gamma'_1<\cdots<\gamma'_n<+\infty}f_{\Gamma'}(\gamma')\int_{-\infty}^{\infty} \frac{\frac{1}{(1-\frac{2j\pi s}{m})^{nm}}-1}{2j \pi s} e^{-2j \pi s \frac{T'  \sigma^2}{N_t}} ds \ d\gamma' 
\\&=\int\limits_{0<\gamma'_1<\cdots<\gamma'_n<+\infty}f_{\Gamma'}(\gamma')\underbrace{\int_{-\infty}^{\infty} f(s) \, ds}_{Q} \, d\gamma'. 
\end{align}
\end{subequations}
In the following we seek to solve $Q$. Note that a pole of order $nm$ exists in the integrand thus the integral $Q$ can be solved using contour integration and its Residue (Res) is
\begin{align}\label{eq:Residueofintegral}
Q=\operatorname{Res}_{z^\ast = \frac{m}{2\pi j}} [f(z)] = \lim_{z\to z^\ast} \frac{1}{(nm-1)!}\biggl(\frac{d}{dz}\biggr)^{nm-1} (z-z^\ast)^{nm}f(z).
\end{align}
The function $f(z)$ can be re-written in the following form 
\begin{subequations}
\begin{align}
&f(z)= \frac{g(z)}{(z-z^{\ast})^{nm}}= \frac{(-1)^{nm}\frac{1-(1-\frac{2\pi jz}{m})^{nm}}{(2\pi j)^{nm}(2\pi jz)}e^{-2\pi jz\frac{T'  \sigma^2}{N_t}}}{\bigl(z-\frac{m}{2\pi j}\bigr)^{nm}},\\
&g(z) = (-1)^{nm}\frac{1-(1-2\pi jz)^{nm}}{(2\pi j)^{nm}(2\pi jz)}e^{-2\pi jz\frac{T'  \sigma^2}{N_t}}.
\end{align}
\end{subequations}
Then after substituting the functions in~\eqref{eq:Residueofintegral}, we can express the integral $I$ as
\begin{align}
Q= \frac{g^{(nm-1)}(z^\ast)}{(nm-1)!}.
\end{align}

\newpage
\bibliography{refs}
\bibliographystyle{IEEEtran}

\end{document}